\RequirePackage{fix-cm}
\documentclass[twocolumn]{svjour3}          
\smartqed  
\usepackage{graphicx}
\usepackage{graphicx}
\usepackage{caption}
\usepackage{subcaption}
\usepackage{amsmath}
\usepackage[utf8]{inputenc}
\usepackage{amssymb}
\usepackage{amsfonts}
\usepackage{epstopdf}
\usepackage[english]{babel}
\usepackage[utf8]{inputenc}
\usepackage[bottom]{footmisc}
\usepackage{wrapfig}
\usepackage[colorinlistoftodos,prependcaption,textsize=tiny]{todonotes}
\usepackage{tabularx,ragged2e,booktabs,caption}
\begin{document}
\newtheorem{mydef}{Definition}
\title{Programming Discrete Distributions with Chemical Reaction Networks\thanks{This research is supported by a Royal Society Research Professorship and by ERC AdG VERIWARE.}}


\author{Luca Cardelli \and Marta Kwiatkowska  \and Luca Laurenti }

\institute{Luca Cardelli \at
             Microsoft Research, Cambridge UK \\
              Department of Computer science, University of Oxford\\
              \email{luca@microsoft.com} \\          
          \and
            Marta Kwiatkowska\at
              Department of Computer science, University of Oxford\\
              \email{marta.kwiatkowska@cs.ox.ac.uk}\\
              \and
            Luca Laurenti \at
              Department of Computer science, University of Oxford\\
              \email{luca.laurenti@cs.ox.ac.uk}\\
}

\date{Received: date / Accepted: date}

\maketitle

\begin{abstract}
We explore the range of probabilistic behaviours that can be engineered with Chemical Reaction Networks (CRNs). 
We give methods to ``program'' CRNs so that their steady state is chosen from some desired target distribution that has finite support in $\mathbb{N}^m$, with $m \geq 1$. Moreover, any distribution with countable infinite support can be approximated with arbitrarily small error under the $L^1$ norm. We also give optimized schemes for special distributions, including the uniform distribution. Finally, we formulate a calculus to compute on distributions that is complete for finite support distributions, and can be compiled to a restricted class of CRNs that at steady state realize those distributions.
\end{abstract}

\section{Introduction}
\label{intro}
Individual cells and viruses operate in a noisy environment and molecular interactions are inherently stochastic. How cells can tolerate and take advantage of noise (stochastic fluctuations) is a question of primary importance. It has been shown that noise has a functional role in cells \cite{eldar2010functional}; indeed, some critical functions depend on the stochastic fluctuations of molecular populations 
and would be impossible in a deterministic setting.
For instance, noise is fundamental for probabilistic differentiation of strategies in organisms, and is a key factor for evolution and adaptation \cite{arkin1998stochastic}. In Escherichia coli, randomly and independently of external inputs, a small sub-population of cells enters a non-growing state in which they can elude the action of antibiotics that can only kill actively growing bacterial cells. Thus, when a population of E. coli cells is treated with antibiotics, the persisted cells survive by virtue of their quiescence before resuming growth  \cite{losick2008stochasticity}. 
This is an example in which molecular systems compute by producing a distribution.  In other cases cells need to shape noise and compute on distributions instead of simply mean values. For example, in \cite{schmiedel2015microrna} the authors show, both mathematically and experimentally, that microRNA confers precision on  the protein expression: it shapes the noise of genes in a way that decreases the intrinsic noise in protein expression, maintaining its expected value almost constant. 
Thus, although fundamentally important, the mechanisms used by cells to compute in a stochastic environment are not well understood.
 
Chemical Reaction Networks (CRNs) with mass action kinetics are a well studied formalism for modelling biochemical systems,
more recently also
used as a formal programming language \cite{Chen2013}. It has been shown that any CRN can be physically implemented by a corresponding DNA strand displacement circuit in a well-mixed solution \cite{soloveichik2010dna}. DNA-based circuits thus have the potential to operate inside cells and control their activity. Winfree and Qian have also shown that CRNs can be implemented on the surface of a DNA nanostructure \cite{qian2014parallel}, 
enabling localized computation and engineering biochemical systems where the molecular interactions occur between few components. When the number of interacting entities is small, the stochastic fluctuations intrinsic in molecular interactions play a predominant role in the time evolution of the system. 
As a consequence, ``programming'' a CRN to provide a particular probabilistic response for a subset of species, for example in response to environmental conditions, is important for engineering complex biochemical nano-devices and randomized algorithms. 
In this paper,  we explore the capacity of CRNs to ``exactly program'' discrete probability distributions. 
That is, we give methods such that the steady state distribution of a CRN can be chosen from some desired target distribution.
We aim  to characterize the probabilistic behaviour that can be obtained, exploring both the capabilities of CRNs for producing distributions and for computing on distributions by composing them.

\vspace{0.5em}
\noindent
$\mathbf{Contributions.}$
We show that at steady state CRNs are able to compute any distribution with finite support in $\mathbb{N}^m$, with $m\geq 1$. We propose an algorithm 
to systematically ``program'' a CRN so that at steady state it produces any given finite support distribution. Moreover, any distribution with countable infinite support can be approximated with arbitrarily small error under the $L^1$ norm. 
The resulting network has a number of reactions linear in the dimension of the support of the distribution and the output is produced monotonically allowing composition. 
Since distributions with large support can result in unwieldy networks, we also give optimised networks for special distributions, including a novel scheme for the uniform distribution.
We formulate a calculus that is complete for finite support distributions, which can be compiled to a restricted class of CRNs that at steady state compute those distributions. The resulting CRNs are generally more compact with respect to the ones derived from direct approach.
The calculus {  is equivalent to the baricentric algebra presented in \cite{mardare2016quantitative}}, and allows for modelling of external influences on the species. 
Our results are of interest for a variety of scenarios in systems and synthetic biology. For example, they can be used to program a biased stochastic coin or a uniform distribution, thus enabling implementation of randomized algorithms and protocols in CRNs.

Preliminary version of this work appeared as \cite{Cardelli2016}. This paper includes an extended description with illustrative examples and proofs of the results. 

\vspace{0.5em}
\noindent
$\mathbf{Related \,\,\,work.}$
It has been shown that CRNs with stochastic semantics are Turing complete, up to an arbitrarily small error \cite{soloveichik2008computation}. If we assume error-free computation, their computational power decreases: they can decide the class of the semi-linear predicates \cite{angluin2007computational} and compute semi-linear functions \cite{chen2014deterministic}. 
A first attempt to model distributions with CRNs can be found in \cite{fett2007synthesizing}, where the problem of producing a single distribution is studied. However, their circuits are approximated and cannot be composed to compute operations on distributions. 

\section{Chemical Reaction Networks}\label{sec-stoch}
A \emph{chemical reaction network (CRN)} $(\Lambda,R)$ is a pair of finite sets, where $\Lambda$ is the set of \emph{chemical species}, $|\Lambda|$ denotes its size, and $R$ is a set of reactions. 
A \emph{reaction} $\tau \in R$ is a triple $\tau=(r_{\tau},p_{\tau},k_{\tau})$, where $r_{\tau} \in  \mathbb{N}^{|\Lambda|}$ is the \emph{source complex}, $p_{\tau} \in  \mathbb{N}^{|\Lambda|}$ is the \emph{product complex} and $k_{\tau} \in \mathbb{R}_{>0} $ is the coefficient associated to the rate of the reaction, where we assume $k_{\tau}=1$ if not specified; $r_{\tau}$ and $p_{\tau}$  represent the stoichiometry of reactants and products.
Given a reaction $\tau_1=(  [1,0,1],[0,2,0],k_1 )$ we often refer to it as $\tau_1 : \lambda_1 + \lambda_3 \, \rightarrow^{k_1}  \,    2\lambda_2 $.
The \emph{net change (or state change)} associated to $\tau$ is defined by $\upsilon_{\tau}=p_{\tau} - r_{\tau}$. 


We assume that the system is well stirred, that is, the probability of the next reaction occurring between two molecules is independent of the location of those molecules, at fixed volume $V$ and temperature. Under these assumptions a \emph{configuration} or \emph{state} of the system $x \in \mathbb{N}^{|\Lambda|}$ is given by the number of molecules of each species. 

A  \emph{chemical reaction system} (CRS) $C=(\Lambda,R,x_0)$ is a tuple where $(\Lambda,R)$ is a CRN and $x_0 \in \mathbb{N}^{|\Lambda|}$ represents its initial condition. 

{ 
\subsection{Stochastic Semantics}
The stochastic semantics of a CRS is given in terms of a continuous time Markov chain (CTMC). Here, we introduce the semantics according to the representation of Markov processes proposed by  Ethier and Kurtz (Theorem 4.1 Chapter 6 \cite{ethier2009markov}). Such representation is equivalent to the classical model described by the Chemical Master Equation, but much more compact.
It allows us to represent the CTMC in terms of stochastic equations, which have a similar structure to the deterministic rate equations.
We illustrate the semantics with the help of Example \ref{Introd}. Below we present Poisson processes, as they will be used in the semantics and in the paper.
A building block of the mathematical models we use in the paper is a \emph{counting process}. Intuitively, a counting process $Y$ is a process such that $Y(t)$ counts the number of times that a particular phenomenon has been observed by time $t$.
\begin{mydef}{(Counting process)}
$Y$ is a counting process if $Y(0)=0$ and $Y$ is constant except for jumps of $+1.$
\end{mydef}
\begin{mydef}{(Poisson process)}
A counting process  $Y$ is a Poisson  process if:
\begin{itemize}
\item Number of observations in disjoint time intervals are independent random variables, that is, $Y(t_k)-$ $Y(t_{k-1}),$ $k \in \mathbb{N}$, are independent random variables.
\item The distribution of $Y(t+\Delta t)-Y(t)$ is independent of $t$.
\end{itemize}
\end{mydef}
\begin{theorem}{(\cite{anderson2015stochastic})}
If $Y$ is a Poisson process, then there exists a constant $\lambda >0$ such that for $t_2>t_1 \in \mathbb{R}_{\geq 0}$ and $k \in \mathbb{N}$ it holds that
$$ Prob(Y(t_2)-Y(t_1)=k)=\frac{(\lambda(t_2-t_1))^k}{k!}e^{-\lambda (t_2 - t_1)}$$
That is, $Y(t_2)-Y(t_1)$ is Poisson distributed with parameter $\lambda(t_2-t_1).$ 
\end{theorem}
If $\lambda=1,$ we call $Y$ a \emph{unit Poisson process}.

\begin{example}\label{Introd}
Consider the CRN described by the following reactions
$$ \tau_1: \lambda_1 + \lambda_2 \to^{k_1} \lambda_1 + \lambda_1;\quad \tau_2: \lambda_1 + \lambda_2 \to^{k_2} \lambda_2 + \lambda_2$$
and let $X(0)\in \mathbb{N}^2$ be the initial condition. Then, the state of the system at time $t\geq 0$ will be given by $X(0)$ plus the number of times that each reaction have fired between $[0,t]$ multiplied by the respective state change vector. That  is,
$$X(t)=X(0) + \begin{pmatrix}
  1  \\
  -1 
 \end{pmatrix} R_{\tau_1}(t) + \begin{pmatrix}
  -1  \\
  1 
 \end{pmatrix} R_{\tau_2}(t)$$  where $R_{\tau_1}(t),R_{\tau_2}(t)$ are counting processes that count the number of times that the particular reaction has fired until time $t$. We now assume that $R_{\tau}$ are independent, unit Poisson processes that depend on the propensity rate of $\tau$. More precisely, $R_{\tau}(t)=Y_{\tau}(\int_0^t \alpha(X(s))ds),$ where $Y_{\tau}(\int_0^t \alpha(X(s))ds)$ is a unit Poisson process with intensity $\int_0^t \alpha(X(s))ds$. Intuitively, $\int_0^t \alpha(X(s))ds$ gives the time interval in which counting events for the unit Poisson process.  
Under this modelling assumptions it holds that \cite{ethier2009markov}
\begin{align*} 
Prob(&Y_{\tau}(\int_0^{t+\Delta t} \alpha_{\tau}(X(s))ds)-\\
&Y_{\tau}(\int_0^{t}\alpha_{\tau}(X(s)ds)>0|\forall s \in [0,t,X(s))\approx\\&\quad \quad \quad \quad \quad \quad \quad \quad \quad \quad \quad \quad \quad \quad \quad \alpha_{\tau}(X(t)) \Delta t . 
\end{align*}

That is, the probability that a reaction $\tau$ happens in the next $\Delta t$, at the first order, is given by the propensity rate of $\tau$ at time $t$ multiplied by $\Delta t$, exactly as in the classical stochastic representation \cite{van1992stochastic} of CRNs. 
At this point, for our model, we can write its stochastic model as
\begin{align*}
 X(t)&=X(0) +\\ &\begin{pmatrix}
  1  \\
  -1 
 \end{pmatrix} Y_{\tau_1}(k_{\tau_1}\int_0^t X_{\lambda_1}(s)X_{\lambda_2}(s)ds) +\\
 & \begin{pmatrix}
  -1  \\
  1 
 \end{pmatrix} Y_{\tau_2}(k_{\tau_2}\int_0^t X_{\lambda_1}(s)X_{\lambda_2}(s)ds).\end{align*} Theorem \ref{CME} below shows that the forward equation associated with the Markov process described in the previous stochastic equation is exactly the \emph{Chemical Master Equation (CME)}.
 \end{example}

 \begin{mydef}
Given a CRS $C=(\Lambda,R,x_0),$ we define its stochastic semantics at time $t$ as
 \begin{align}\label{StochasticSemantics}
 X^C(t)=x_0 + \sum_{\tau \in R} \upsilon_{\tau} Y_{\tau}(\int_0^t \alpha_{\tau}(X^C(s)ds)) 
 \end{align}
 where $Y_{\tau} $ are unit Poisson processes, independent of each other. 
 \end{mydef}

\begin{theorem}{\cite{ethier2009markov}}\label{CME}
Let $C=(\Lambda,R,x_0)$ be a CRS and $X^C$ be the stochastic process as defined in Equation \eqref{StochasticSemantics}. Define $Prob(X^C(t)=x|X^C(0)=x_0)=P^C(t)(x)$. Assume that, for each $\tau \in R$ and $t\in \mathbb{R}_{\geq 0},$
$X^C(t)< \infty,$ then 
\begin{align}
&\frac{d P^C(t)(x)}{dt}=\nonumber\\
&\quad\sum_{\tau \in R} P^C(t)(x-\upsilon_{\tau})\alpha_\tau(X^C(t))-P^C(t)(x)\alpha_\tau(X^C(t)).\label{CMEEq}
\end{align}
\end{theorem}
}$P^C(t)(x)$ represents the transient evolution of $X^C$, and can be calculated exactly by solving directly the Chemical Master Equation or by approximation techniques \cite{laurenti2015stochastic,laurenti2016stochastic,bortolussi2016approximation}. 
\begin{mydef}
The steady state distribution (or limit distribution) of $X^C$ is defined as $\pi^C= \lim_{t \to \infty} P^C(t).$
\end{mydef}
When clear from the context, we omit the superscript indicating the CRN and simply write $\pi$ instead of $\pi^C$.  
$\pi$ calculates the percentage of time, in the long-run, that $X$ spends in each state $x \in S$.
If $S$ is finite, 
then the above limit distribution always exists and is unique \cite{kwiatkowska2007stochastic}. 
In this paper we focus on discrete distributions, and will sometimes conflate the term distribution with probability mass function, defined next.
\begin{mydef}\label{prob mass func-defn}
Suppose that $M: S \rightarrow \mathbb{R}^m$ with $m >0$ is a discrete random variable defined on a countable sample space $S$. Then the probability mass function (pmf) $f: \mathbb{R}^m \rightarrow [0, 1] $ for $M$ is defined as
    $f(x) = Prob(s \in S  \mid  M(s) = x).  $
\end{mydef}
For a pmf $\pi : \mathbb{N}^{m} \rightarrow [0,1]$ we call $J=\{y\in \mathbb{N}^{m}|\pi(y)\neq 0\}$ the support of $\pi$.  A pmf is always associated to a discrete random variable whose distribution is described by the pmf. Sometimes, when we refer to a pmf, we imply the associated random variable. Given two pmfs $f_1$ and $f_2$ with values in $\mathbb{N}^m$, $m>0$, we define the $L^1$ norm (or distance) between them as $d_1(f_1,f_2)=\sum_{n\in \mathbb{N}^m}(|f_1(n)-f_2(n)|)$.
Note that, as $f_1,f_2$ are pmfs, then $d_1(f_1,f_2)\leq 2$. 
It is worth stressing that, given the CTMC $X$, for each $t \in \mathbb{R}_{\geq 0}$, $X(t)$ is a random variable defined on a countable state space. As  a consequence, its distribution is given by a pmf. Likewise, the limit distribution of a CTMC, if it exists, is a pmf.

\begin{mydef}
Given $C=(\Lambda,R)$ and $\lambda \in \Lambda$,  we define $\pi_{\lambda} (k)=\sum_{ \{x\in S |x(\lambda)=k\}} \pi(x)$ as the probability that for $t \rightarrow \infty$, in $X^C$, there are $k$ molecules of $\lambda$.
\end{mydef}
$\pi_{\lambda}$ is a pmf representing the steady state distribution of species $\lambda$.

\section{On computing finite support distributions with CRNs}
We now show that, for a pmf with { finite} support in $\mathbb{N}$, we can always build a CRS such that, at steady state (i.e. for $t \rightarrow \infty$) the random variable representing the molecular population of a given species in the CRN {  is equal to that distribution. Such result allows us to approximate any distribution with countable infinite support with arbitrarily small error under the $L^1$ norm}. 
The result is then generalised to distributions with domain in $\mathbb{N}^m$, with $m \geq 1$. The approximation is exact in case of finite support.
\subsection{Programming pmfs}
\begin{mydef}\label{CRN univ}
Given $f: \mathbb{N} \rightarrow [0,1]$ with finite support $J=(z_1,...,z_{|J|})$ such that $\sum_{i=1}^{|J|}f(z_i)=1$, we define the CRS $C_{f}=(\Lambda,R,x_0)$ as follows. $C_{f}$ is composed of $2|J|$ reactions and $2|J|+2$ species.
For any $z_i \in J$ we have two species $\lambda_i,\lambda_{i,i} \in \Lambda$ such that $x_0(\lambda_i)=z_i$ and $x_0(\lambda_{i,i})=0$. Then, we consider a species $\lambda_z \in \Lambda$ such that $x_0(\lambda_z)=1$, and the species $\lambda_{out} \in \Lambda$, which represents the output of the network and such that  $x_0(\lambda_{out})=0$.
 For every $z_i \in J$, $R$ has the following two reactions: $\tau_{i,1}: \lambda_z \rightarrow^{f(z_i)}  \lambda_{i,i}$ and $\tau_{i,2}: \lambda_i + \lambda_{i,i}  \rightarrow \lambda_{out} + \lambda_{i,i}$.
\end{mydef}
\begin{example}\label{ex-1}
Consider the probability mass function $f: \mathbb{N}\rightarrow [0,1]$ defined as
$f(y)= 
\left\{ 
  \begin{array}{l l}
    \frac{1}{6}, \, \, \, \,\,\, \text{if $y=2$}\\
    \frac{1}{3}, \, \, \, \,\,\, \text{if $y=5$}\\
    \frac{1}{2}, \, \, \,  \,\,\, \text{if $y=10$}\\
    0, \, \, \, \,\,\, \text{otherwise}\\
  \end{array} \right.
 $.
Let $\Lambda=\{\lambda_{1}, \lambda_{2}, \lambda_{3},$ $ \lambda_z,  \lambda_{1,1},  \lambda_{2,2},   \lambda_{3,3}, \lambda_{out} \}$, then we build the CRS $C=(\Lambda,R,x_0)$ following Definition \ref{CRN univ}, where
$R$ is given by the following set of reactions:
\[\lambda_z \rightarrow^{\frac{1}{6}}   \lambda_{1,1};\,\,\,\,\, \lambda_z \rightarrow^{\frac{1}{3}}    \lambda_{2,2};\,\,\,\,\,
  \lambda_z \rightarrow^{\frac{1}{2}}    \lambda_{3,3};\]
\[\lambda_{1} + \lambda_{1,1} \rightarrow^{1} \lambda_{1,1} + \lambda_{out};\,\,\,\,\,
\lambda_{2} + \lambda_{2,2} \rightarrow^{1} \lambda_{2,2} + \lambda_{out};\,\,\,\,\,\]
 \[\lambda_{3} + \lambda_{3,3} \rightarrow^{1} \lambda_{3,3} + \lambda_{out}.\]
The initial condition $x_0$ is 
$ x_0 (\lambda_{out})=x_0(\lambda_{1,1})=x_0(\lambda_{2,2})$ $=x_0(\lambda_{3,3})=0;$
$x_0(\lambda_{1})=2; \, x_0(\lambda_{2})=5;$  
$x_0(\lambda_{3})=10;$ $x_0(\lambda_z)=1.$
Theorem \ref{th:SingleCase} ensures $\pi_{\lambda_{out}}=f$.
\end{example}
\begin{theorem}\label{th:SingleCase}
Given a pmf  $f: \mathbb{N}  \rightarrow [0, 1] $ with finite support $J$, the CRS $C_{f}$ as defined in Definition \ref{CRN univ} is such that $\pi^{C_f}_{\lambda_{out}}=f$.
\end{theorem}
\begin{proof}
Let $J=( z_1,..,z_{|J|} )$ be the support of $f$, and $|J|$ its size. Suppose  $|J|$ is finite,
then the set of reachable states from $x_0$ is finite by construction 
and the limit distribution of $X^{C_f}$, the induced CTMC, exists.
By construction, in the initial state $x_0$ only reactions of type $\tau_{i,1}$ can fire, and the probability that a specific $\tau_{i,1}$ fires first is exactly:
\begin{align*}
\frac{\alpha_{\tau_{i,1}}(x_0)}{\sum_{j=1}^{|J|}\alpha_{\tau_{j,1}}(x_0)}=&\frac{f(z_i) \cdot 1 }{\sum_{j=1}^{|J|}f(z_j) \cdot 1 } =\\
&\quad \quad \quad \frac{f(z_i) }{\sum_{j=1}^{|J|}f(z_j) }=\frac{f(z_i)}{1}=f(z_i)
\end{align*}
Observe that the firing of the first reaction uniquely defines the limit distribution of $X^{C_{f}}$, because $\lambda_z$ is consumed immediately and only reaction $\tau_{i,2}$ can fire, with no race condition, until $\lambda_i$ are consumed.
This implies that at steady state $\lambda_{out}$ will be equal to $x_0(\lambda_{i})$, and this happens with probability $f(x_0(\lambda_{i}))$. Since $x_0(\lambda_{i})=z_i$ for $i \in [1,|J|]$, we have $\pi_{\lambda_{out}}^{C_f}=f$.
\hfill $\square$
\end{proof}
Then, we can state the following corollary of Theorem \ref{th:SingleCase}.
\begin{corollary}\label{univer}
Given a pmf  $f: \mathbb{N}  \rightarrow [0, 1] $ with countable support $J$, we can always find a finite CRS $C_{f}$ such that $\pi^{C_f}_{\lambda_{out}}=f$ with arbitrarily small error under the $L^1$ norm.
\end{corollary}
\begin{proof}
Let $J=\{z_1,...,z_{|J|} \}$. Suppose $J$ is (countably) infinite, that is, $|J| \rightarrow \infty$. Then, we can always consider an arbitrarily large but finite number of points in the support, such that the probability mass lost is arbitrarily small, and applying Definition \ref{CRN univ} on this finite subset of the support we have the result.

In order to prove the result consider the function $f'$ with support $J'=\{z_1,...,z_{k} \}$, $k\in \mathbb{N}$, such that $f(z_i)=f'(z_i)$, for all $i \in \mathbb{N}_{\leq k}$.
Consider the series $\sum_{i=1}^{\infty}f(n)$. This is an absolute convergent series by definition of pmf. Then, we have that $\lim_{i\rightarrow \infty}f(i)=0$ and,
for any $\epsilon > 0$, we can choose some $\kappa_\varepsilon \in \mathbb{N}$, such that:
\begin{align*}
\forall k>\kappa_\varepsilon \quad |\sum_{i=1}^k f'(i)-\sum_{i=1}^\infty f(i)| < \frac{\epsilon}{2}.
\end{align*}
This implies that for $k>\kappa_\varepsilon$ given $f'_k=\sum_{i=1}^k f'(i)$ we have, $d_1(f'_k,f)<\epsilon$.
\hfill $\square$


\end{proof}
\noindent
The following remark shows that the need for precisely tuning the value of reaction rates in Theorem \ref{th:SingleCase} can be dropped by introducing some auxiliary species.
\begin{remark}\label{SingAlt}
In practice, tuning the rates of a reaction can be difficult  or impossible. However, it is possible to modify the CRS derived using Definition \ref{CRN univ} in such a way the probability value is not encoded in the rates, and we just require that all reactions have the same rates.
We can do that by using some auxiliary species $\Lambda_c=\{\lambda_{c_1},\lambda_{c_2},...,\lambda_{c_{|\Lambda_c|}} \}$. Then, the reactions $\tau_{i,1}$ for $i\in [1,J]$ become $\tau_{i,1}: \lambda_z+\lambda_{c_i} \rightarrow^{k}  \lambda_{i,i}$, for $k\geq 0$, instead of $\tau_{i,1}: \lambda_z \rightarrow^{f(y_i)}  \lambda_{i,i}$, as in the original definition. The initial condition of $\lambda_{c_i}$ is $x_0(\lambda_{c_i})=f(y_i)\cdot L$, where $L \in \mathbb{N}$ is such that for $j \in [1,|J|]$ and $J=\{ z_1,...,z_{|J|} \}$ we have that $f(z_j) \cdot L$ is a natural number, assuming all the $f(z_j)$ are rationals.
\end{remark}

\begin{remark}\label{ExternalInfluence}
In biological circuits the probability distribution of a species may depend on some external conditions. For example, the \emph{lambda Bacteriofage} decides to lyse or not to lyse with a probabilistic distribution based also on environmental conditions \cite{arkin1998stochastic}. 
Programming similar behaviour is possible by extension of Theorem \ref{th:SingleCase}.
For instance, suppose, we want to program a switch that with rate $50+Com$ goes to state $O_1$, and with rate $5000$ goes to a different state $O_2$, where $Com$ is an external input. 
To program this logic we can use the following reactions: $\tau_{1,1} : \lambda_z + \lambda_{c_1} \rightarrow^{k_1} \lambda_{O_1}$ and $\tau_{1,2} : \lambda_z + \lambda_{c_2} \rightarrow^{k_1} \lambda_{O_2}$, where $\lambda_{O_1}$ and $\lambda_{O_2}$ model the two logic states, initialized at $0$. The initial condition $x_0$ is such that $x_0(\lambda_z)=1$, $x_0(\lambda_{c_1})=50$ and $x_0(\lambda_{c_2})=5000$. Then, we add the following reaction $Com \rightarrow^{k_2} \lambda_{c_1}$. It is easy to show that if $k_2 \gg k_1$ then we have the desired probabilistic behaviour for any initial value of $Com \in \mathbb{N}$. 
This may be of interest also for practical scenarios in synthetic biology, where for instance the behaviour of synthetic bacteria needs to be externally controlled \cite{anderson2006environmentally}; and, if each bacteria is endowed with a similar logic, then, by tuning $Com$, at the population level, it is possible to control the fraction of bacteria that perform this task.
\end{remark}
In the next theorem we generalize to the multidimensional case.
\begin{theorem}\label{th:Multi}
Given $f: \mathbb{N}  ^{m}  \rightarrow [0, 1] $ with $m\geq 1$ such that $\sum_{i\in \mathbb{N}^m}f(i)=1$, then there exists a CRS $C=(\Lambda,R,x_0)$ such that the joint limit distribution of $(\lambda_{out_1},$ $\lambda_{out_2},...,\lambda_{out_m}) \in \Lambda$ approximates $f$ with arbitrarily small error under the $L^1$ distance. The approximation is exact if the support of $f$ is finite.
\end{theorem}
To prove this theorem we can derive a CRS similar to that in the uni-dimensional case. The firing of the first reaction can be used to probabilistically determine the value at steady state of the $m$ output species, using some auxiliary species.

{ 
\begin{example}\label{ex-2}
Consider the following probability mass function 
\[ f(y_1,y_2)= \left\{ 
  \begin{array}{l l}
    \frac{1}{6}, \,\,\, \, \, \,  \text{if $y_1=3$ and $y_2=1$}\\
    \frac{1}{3}, \,\,\,\, \, \,  \text{if $y_1=3$ and $y_2=2$}\\
    \frac{1}{2}, \, \, \,  \,\,\, \text{if $y_1=1$ and $y_2=5$}\\
    0, \, \, \, \,\,\, \text{otherwise}\\
  \end{array} \right.
 \]
we present  the CRS $C=(\Lambda,R,x_0)$ that according to its stochastic semantics, for $\lambda_{out_1}, \lambda_{out_2} \in \Lambda$ yields the steady-state distribution $\pi_{\lambda_{out_1},\lambda_{out_2}}$, joint limit distribution of $\lambda_{out_1},\lambda_{out_2}$, exactly equal to $f$.
Let $\Lambda=$ $\{ \lambda_z, \lambda_a, \lambda_b,$ $ \lambda_c, \lambda_{1,1}, \lambda_{1,2}  \lambda_{2,1}, \lambda_{2,2},   \lambda_{3,1}, $ $\lambda_{3,2} \lambda_{out_1},\lambda_{out_2} \}$ and $R$ given by the following set of reactions:
\begin{align*}
&\tau_1:  \lambda_z \rightarrow^{\frac{1}{6}}    \lambda_{a};\quad
\tau_2:  \lambda_z \rightarrow^{\frac{1}{3}}   \lambda_{b};\quad
\tau_3: \lambda_z \rightarrow^{\frac{1}{2}}   \lambda_{c};\\
&\tau_4: \lambda_{1,1} + \lambda_{a} \rightarrow^{1} \lambda_{a} + \lambda_{out_1};\\
&\tau_5: \lambda_{1,2} + \lambda_{a} \rightarrow^{1} \lambda_{a} + \lambda_{out_2};\\
&\tau_6: \lambda_{2,1} + \lambda_{b} \rightarrow^{1} \lambda_{b} + \lambda_{out_1};\\
&\tau_7: \lambda_{2,2} + \lambda_{b} \rightarrow^{1} \lambda_{b} + \lambda_{out_2};\\
&\tau_8: \lambda_{3,1} + \lambda_{c} \rightarrow^{1} \lambda_{c} + \lambda_{out_1};\\&\tau_9: \lambda_{3,2} + \lambda_{c} \rightarrow^{1} \lambda_{c} + \lambda_{out_2};
\end{align*}
The initial condition $x_0$ is such that:
$$  x_0(\lambda_z)=1;$$
$$ x_0(\lambda_{1,1})=3;\,  x_0(\lambda_{1,2})=1;\,  x_0(\lambda_{2,1})=3; $$
$$ x_0(\lambda_{2,2})=2;  \, x_0(\lambda_{3,1})=1; \, x_0(\lambda_{3,2})=5;$$
and all other species mapped to zero.
The set of reachable states from $x_0$ is finite so the limit distribution exists. The firing of the first reaction uniquely determines the steady state solution. $x_0(\lambda_{i,1})$ and $x_0(\lambda_{i,2})$ for $i \in [1,3]$ are exactly the value of $\lambda_{out_1}$ and $\lambda_{out_2}$ at steady state if the first reaction to fire is $\tau_i$; this happens with probability $f(x_0(\lambda_{i,1}),x_0(\lambda_{i,2}))$. Therefore, we have that, at steady state, the joint distribution of $\lambda_{out_1}$ and $\lambda_{out_2}$  equals $f$. 

\end{example}


}

\subsection{Special distributions}\label{SpecSect}
For a given pmf the number of reactions of the CRS derived from Definition \ref{CRN univ} is linear in the dimension of its support. As a consequence, if the support is large then the CRSs derived using Theorems \ref{th:SingleCase} and \ref{th:Multi} can be unwieldy. In the following we show three optimised CRSs to calculate the Poisson, binomial and uniform distributions. These CRNs are compact and applicable in many practical scenarios. However, using  Definition \ref{CRN univ} the output is always produced monotonically. In the circuits below this does not happen, but, on the other hand, the gain in compactness is substantial.
The first two circuits have been derived from the literature, while the CRN for the uniform distribution is new.
\subsubsection{Poisson distribution}
The main result of \cite{anderson2010product}  guarantees that all the CRNs that respect some conditions (weakly reversible, deficiency zero and irreducible state space, see \cite{anderson2010product}) have a distribution given by the product of Poisson distributions.
As a particular case, we consider the following CRS composed of only one species $\lambda$ and the following two reactions
$  \tau_1 : \emptyset \rightarrow^{k_1} \lambda; \, \tau_2:   \lambda \rightarrow^{k_2} \emptyset.$
Then, at steady state, $\lambda$ has a Poisson distribution with expected value $\frac{k_1}{k_2}$.
\subsubsection{Binomial distribution}
We consider the network introduced in \cite{anderson2010product}. The CRS is composed of two species, $\lambda_1$ and $\lambda_2$, with initial condition $x_0$ such that $x_0(\lambda_1)+x_0(\lambda_2)=K$ and the following set of reactions:
$
\tau_1:\lambda_1 \rightarrow^{k_1} \lambda_2; 
\tau_2:\lambda_2 \rightarrow^{k_2} \lambda_1.$
As shown in \cite{anderson2010product}, $\lambda_1$ and $\lambda_2$ at steady state have a binomial distribution such that:
$
\pi_{\lambda_1}(y)=(\frac{K}{y}){c_1}^{y}(1-c_1)^{K-y}
 \text{   and    }
\pi_{\lambda_2}(y)=(\frac{K}{y}){c_2}^{y}(1-c_2)^{K-y} .
$ 
\noindent
\subsubsection{Uniform distribution}
The following CRS computes the uniform distribution over the sum of the initial number of molecules in the system, independently of the initial value of each species. It has species $\lambda_1$ and $\lambda_2$ and reactions:
\begin{align*}
    &\tau_1:\lambda_{1} \rightarrow^{k}   \lambda_{2}; \quad
 \tau_2:\lambda_{2} \rightarrow^{k} \lambda_{1};\\
&\tau_3:\lambda_{1} + \lambda_{2} \rightarrow^{k} \lambda_{1} + \lambda_{1};\quad
 \tau_4:\lambda_{1} + \lambda_{2} \rightarrow^{k} \lambda_{2} + \lambda_{2}
 \end{align*}
For $k >0$,
$\tau_1$ and $\tau_2$ implement the binomial distribution. These are combined with $\tau_3$ and $\tau_4$, which implement a Direct Competition (DC) system \cite{cardelli2012cell}. DC has a bimodal limit distribution in $0$ and in $K$, where $x_0(\lambda_1)+x_0(\lambda_2)=K$, with $x_0$ initial condition.
This network, surprisingly, according to the next theorem, at steady state produces a distribution which varies uniformly between $0$ and $K$.
\begin{theorem}
Let $x_0(\lambda_1)+x_0(\lambda_2)=K \in \mathbb{N}$. Then,  the CRS described above has the following steady state distribution for $\lambda_1$ and $\lambda_2$:
\[ \pi_{\lambda_1}(y)=\pi_{\lambda_2}(y)=\left\{ 
  \begin{array}{l l}
    \frac{1}{K+1}, \,\,\, \, \, \,  \text{if $y \in [0,K]$}\\
    0, \, \, \, \,\,\, \text{otherwise}\\
  \end{array} \right. .
 \]
\end{theorem}
\begin{proof}
We consider a general initial condition $x_0$ such that $x_0(\lambda_1)=K-M$ and $x_0(\lambda_2)=M$ for  $0 \leq M \leq K$ and $K,M \in \mathbb{N}$. Because any reaction has exactly $2$ reagents and $2$ products, we have the invariant that for any configuration $x$ reachable from $x_0$ it holds that $x(\lambda_1)+x(\lambda_2)=K$. Figure \ref{fig:CTMC} plots 
the CTMC semantics of the system.
\begin{figure*}
	\centering
	\includegraphics[scale=0.25]{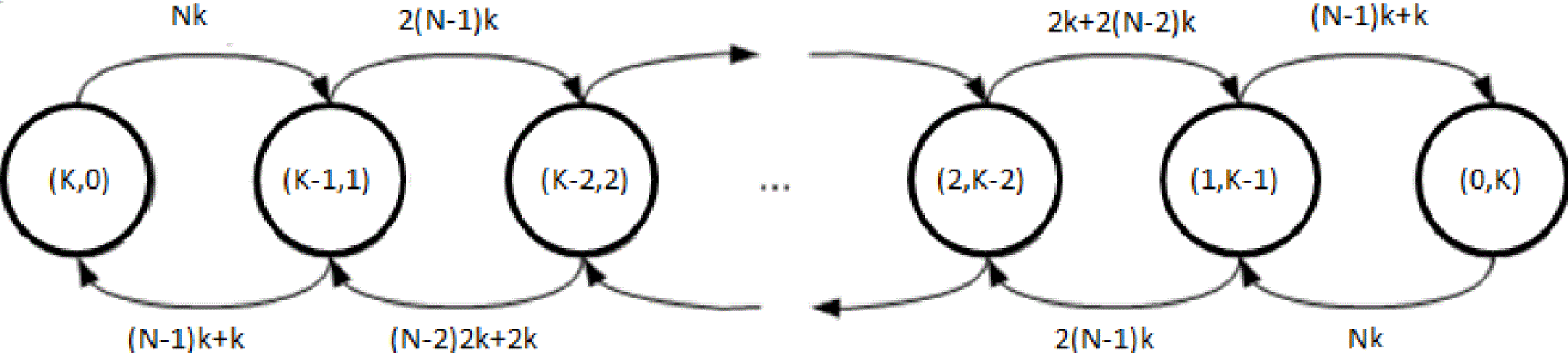} 
	\caption{The figure shows the CTMC induced by the CRS implementing the uniform distribution for initial condition $x_0$ such that $x_0(\lambda_1)+x_0(\lambda_2)=K$.  }
	\label{fig:CTMC}
\end{figure*}%
For any fixed $K$ the set of reachable states from any initial condition in the induced CTMC is finite (exactly $K$ states are reachable from any initial condition) and irreducible. Therefore, the steady state solution exists, is unique and independent of the initial conditions. To find this limit distribution we can calculate $Q$, the infinitesimal generator of the CTMC, and then solve the linear equations system $\pi Q=0$, with the constraint that $\sum_{i \in [0,K]}\pi_{i}=1$, where $\pi_{i}$ is the $i$th component of the vector $\pi$, as shown in \cite{kwiatkowska2007stochastic}. Because the CTMC we are considering is irreducible, this is equivalent to solving the balance equations with the same constraint. The resulting $\pi$ is the steady state distribution of the system. 

We consider $3$ cases, where $(K-j,j)$ for $j \in [0,K]$ represents the state of the system in terms of molecules of $\lambda_1$ and $\lambda_2$.
\begin{itemize}

\item{Case $j=0$.} For the state $(K,0)$, whose limit distribution is defined as $\pi (K,0),$ we have the following balance equation:
\[
-\pi(K,0) Kk+\pi(K-1,1)[(K-1)k+k]=0  \implies
\]
\[
\pi(K,0)=\pi(K-1,1).
\]
\item{Case $j \in [1,K-1]$.}
In Figure \ref{fig:CTMC} we see that the states and the rates follow a precise pattern: every state is directly connected with only two states and for any transition the rates depend on two reactions, therefore we can consider the balance equations for a general state $(K-j,j)$ for $j\in [1,K-1]$ (for the sake of a lighter notation instead of $\pi(K-j,j)$ we write $\pi^{j}$):
\begin{align*}
\pi^{j-1}&[ K +1 -j+(K+1-j)(j-1)]- \\
&
\pi^{j}[2(K-j)j+j+K-j]+\\
&\pi^{j+1}[j+1+(K-j-1)(j+1)] =0 \\
&\quad\quad\quad\quad\quad\quad\quad \implies\\
\pi^{j-1}&[ Kj-j^2+j]-\\
&\pi^{j}[2Kj-2j^2 + K]+\\
&\pi^{j+1}[Kj+K-j^2-j] = 0
\end{align*}
It is easy to verify that if $\pi^{j-1}= \pi^j = \pi^{j+1}$ then the equation is proved. 
\item{Case $j=K$.} The case for the state $(0,K)$ is similar to the case $(K,0)$. 
\end{itemize}
We have shown that each reachable state has equal probability at steady state for any possible initial condition. Therefore, because $\sum_{i=0}^K \pi^i = 1$ and $\pi_{\lambda_i}(y)=$ $\sum_{ x_{j} \in S|x_j(\lambda_i)=y} \pi^j $ for $y \geq 0$,  we have that for both $\lambda_1$ and $\lambda_2$  \[ \pi_{\lambda_1}(y)=\pi_{\lambda_2}(y)=\left\{ 
  \begin{array}{l l}
    \frac{1}{K+1}, \,\,\, \, \, \,  \text{if $y \in [0,K]$}\\
    0, \, \, \, \,\,\, \text{otherwise}\\
  \end{array} \right.
 \]
\hfill $\square$
\end{proof}

\section{Calculus of limit distributions of CRNs}

In the previous section we have shown that CRNs are able to program any pmf on $\mathbb{N}$. We now define a calculus to compose and compute on pmfs.  We show it is complete with respect to finite support pmfs on $\mathbb{N}$. The calculus we present is a left-invariant baricentric algebra \cite{mardare2016quantitative}. Then, we define a translation of this calculus into a restricted class of CRNs. We prove the soundness of such a translation, which thus yields an abstract calculus of limit distributions of CRNs. For simplicity, in what follows we consider only pmfs with support in $\mathbb{N}$, but the results can be generalised to the multi-dimensional case. 
\begin{mydef}{(Syntax).}\label{LangSynt} The syntax of formulae of our calculus is given by
\[P:=\,(P+P) \, | \, min(P,P) \, | \, k \cdot P \, | \, (P)_D:P \, | \, one \, | \, zero
\]
\[D := p \,|\, p \cdot c_i  + D \, \]
where $k \in \mathbb{Q}_{\geq 0}$, $p \in \mathbb{Q}_{[0,1]}$ are rational and $V=\{c_1,...,$ $c_n \}$ is a set of variables with values in $\mathbb{N}$.
\end{mydef}
A formula $P$ denotes a pmf that can be obtained as a sum, minimum, multiplication by a rational, or convex combination of pmfs $one$ and $zero$.
Given a formula $P$, 
variables $V=\{c_1,...,c_n \}$, called \emph{environmental inputs}, model the influence of external factors on the probability distributions of the system. $V(P)$ represents the variables in $P$. An \emph{environment} $E: V \rightarrow \mathbb{Q}_{[0,1]} $ is a partial function which maps each input $c_i$ to its valuation normalized to $[0,1]$.
Given a formula $P$ and an environment $E$, where $V(P)\subseteq dom(E)$, with $dom(E)$ domain of $E$, we define its semantics, $[\![P]\!]_E$, 
as a pmf 
(the empty environment is denoted as $\emptyset$). $D$ expresses a summation of valuations of inputs $c_i$ weighted by rational probabilities $p$, which evaluates to a rational $[\![D]\!]_E$ for a given environment. We require that, for any $D$, the sum of $p$ coefficients in $D$ is in $[0,1]$. This ensures that $0\leq [\![D]\!]_E \leq 1$. The semantics is defined inductively as follows, where the operations on pmfs are defined in Section \ref{op}. 
\begin{mydef}{(Semantics).} Given formulae $P,$ $P_1,$ $P_2$ and an environment $E$, such that $V(P)\cup V(P_1)\cup V(P_2)$ $\subseteq dom(E)$, we define
\begin{align*}
   &[\![one]\!]_E=\pi_{one}\quad \quad \quad[\![ zero]\!]_E=\pi_{zero}\\
   & [\![P_1+P_2]\!]_E=[\![P_1]\!]_E + [\![P_2]\!]_E \\
   & [\! [min(P_1,P_2)]\!]_E=min([\![P_1]\!]_E,[\![P_2]\!]_E) \\
   & [\![k\cdot P]\!]_E=\frac{k_1\cdot ([\![P]\!]_E)}{k_2}\,\,\,\, \text{for $k=\frac{k_1}{k_2}$ and $k_1,k_2 \in \mathbb{N}$}\\
   & [\![(P_1)_D:(P_2)]\!]_E=([\![P_1]\!]_E)_{[\![D]\!]_E} : ([\![P_2]\!]_E) \\
   & [\![p]\!]_E=p \\
   & [\![p \cdot c_i  + D]\!]_E=p \cdot E(c_i)  + ([\![D]\!]_E)
\end{align*} 
where $$\pi_{one}(y)=\left\{ 
 \begin{array}{l l}
    1, \,   \text{if $y=1$ }\\
    0, \,  \text{otherwise}\\
  \end{array} \right. \text{, } \pi_{zero}(y)=\left\{ 
 \begin{array}{l l}
    1, \,   \text{if $y=0$ }\\
    0, \,  \text{otherwise}\\
  \end{array} \right..$$
\end{mydef}
To illustrate the calculus, consider the Bernoulli distribution with parameter $p\in \mathbb{Q}_{[0,1]}$. We have $bern^p=(one)_p:zero$, where $[\![bern^p]\!]_{\emptyset}(y)=\{p \,\,\,   \text{if $y=1$};1-p \,\,\,  \text{if $y=0$};0 \,\,\, \text{otherwise}\}$.
The binomial distribution can be obtained as a sum of $n$ independent Bernoulli distributions of the same parameter. Given a random variable with a binomial distribution with parameters $(n,p)$, if $n$ is sufficiently large and $p$ sufficiently small then this approximates a Poisson distribution with parameter $n\cdot p$.
\subsection{Operations on distributions}\label{op}
In this section, we define a set of operations on pmfs needed to define the semantics of the calculus. We conclude the section by showing that these operations are sufficient to represent pmfs with finite support in $\mathbb{N}$.

\begin{mydef}\label{Op-defn}
  Let $\pi_1:\mathbb{N} \rightarrow [0,1]$,  $\pi_2:\mathbb{N} \rightarrow [0,1] $ be two pmfs. Assume $p \in \mathbb{Q}_{[0,1]}$, $y \in \mathbb{N}$, $k_1 \in \mathbb{N}$ and $k_2 \in \mathbb{N}_{>0}$, then  we define the following operations on pmfs:
\begin{itemize}
\item The sum or convolution of $\pi_1$ and $\pi_2$ is defined as 
$$(\pi_1 + \pi_2)(y)=\sum_{ (y_i,y_j) \in {\mathbb{N}\times \mathbb{N}} \, s.t. \, y_i+y_j=y} \pi_1(y_i)\pi_2(y_j).
 $$
 \item The minimum of $\pi_1$ and $\pi_2$ is defined as
\begin{align*}
min(\pi_1,&\pi_2)(y)=\\ 
&\sum_{(y_i,y_j) \in \mathbb{N}\times \mathbb{N}\,s.t.\,min(y_i,y_j)=y} \pi_{1}(y_i)\pi_{2}(y_j).
 \end{align*}
\item  The multiplication of $\pi_1$ by the constant $k_1$ is defined as
$$
   (k_1 \pi_1) (y)= \left\{ 
  \begin{array}{l l}
     \pi_1(\frac{y}{k_1}), \,\,\, \, \, \,  \text{if $\frac{y}{k_1} \in \mathbb{N}$ }\\
    0, \, \, \, \,\,\, \text{otherwise}\\
  \end{array} \right.
    $$
\item  The division of $\pi_1$ by the constant $k_2$ is defined as 
    $$\frac{\pi}{k_2}(y)=\sum_{y_i \in \mathbb{N}\,s.t.\, y=\lfloor y_i / k_2 \rfloor}\pi(y_i).$$
\item  The convex combination of $\pi_1$ and $\pi_2$, for $y \in \mathbb{N}$, is defined as
$$( (\pi_1)_p : (\pi_2) )(y)= p\pi_1 (y)+ (1-p) \pi_2 (y)$$.

\end{itemize}

\end{mydef}

{               
\begin{example}\label{ex-COm}
Consider the following pmf $\pi_1: \mathbb{N}\rightarrow [0,1]$
\[ \pi_{1}(y_1)= \left\{ 
  \begin{array}{l l}
    \frac{1}{6}, \,\,\, \, \, \,  \text{if $y_1=3$ }\\
    \frac{5}{6}, \, \, \,  \,\,\, \text{if $y_1=0$ }\\
    0, \, \, \, \,\,\, \text{otherwise}\\
  \end{array} \right.
 \]
and the following pmf $\pi_2: \mathbb{N}\rightarrow [0,1]$
\[ \pi_{2}(y_2)= \left\{ 
  \begin{array}{l l}
    \frac{1}{2}, \,\,\, \, \, \,  \text{if $y_2=5$ }\\
    \frac{1}{2}, \, \, \,  \,\,\, \text{if $y_2=1$ }\\
    0, \, \, \, \,\,\, \text{otherwise}\\
  \end{array} \right.
 \]
Then the sum of $\pi_1$ and $\pi_2$ is: 
\[ (\pi_1 + \pi_2)(y)= \left\{ 
  \begin{array}{l l}
    \frac{1}{12}, \,\,\, \, \, \,  \text{if $y=8$ }\\
    \frac{5}{12}, \,\,\, \, \, \,  \text{if $y=5$ }\\
    \frac{1}{12}, \,\,\, \, \, \,  \text{if $y=4$ }\\
    \frac{5}{12}, \, \, \,  \,\,\, \text{if $y=1$ }\\
    0, \, \, \, \,\,\, \text{otherwise}\\
  \end{array} \right.
 \]
\end{example}
\begin{example}\label{ex-MinDef}
Consider the pmfs $\pi_1$ and $\pi_2$ of Example \ref{ex-COm} then  \[ min(\pi_1,\pi_2)(y)= \left\{ 
  \begin{array}{l l}
    \frac{1}{12}, \,\,\, \, \, \,  \text{if $y=3$ }\\
    \frac{1}{12}, \,\,\, \, \, \,  \text{if $y=1$ }\\
    \frac{5}{6}, \,\,\, \, \, \,  \text{if $y=0$ }\\
    0, \, \, \, \,\,\, \text{otherwise}\\
  \end{array} \right.
 \]
\end{example}

\begin{example}\label{ex-Sub}
Consider the pmf $\pi_2$ of Example \ref{ex-COm}, then  \[ 2\pi_{2}(y)= \left\{ 
  \begin{array}{l l}
    \frac{1}{2}, \,\,\, \, \, \,  \text{if $y=10$ }\\
    \frac{1}{2}, \,\,\, \, \, \,  \text{if $y=2$ }\\
    0, \, \, \, \,\,\, \text{otherwise}\\
  \end{array} \right.
 \]
\end{example}

}
\begin{example}\label{pmf}
Consider the following formula 
\[P_1=(one)_{0.001\cdot c + 0.2}:(4\cdot one)+(2 \cdot one)_{0.4}:(3\cdot one),\] with set of environmental variables $V=\{c \}$ and an enviroment $E$ such that $V(P_1)\subseteq dom(E)$.
Then, according to Definition \ref{Op-defn} we have that 
\[[\![P_1]\!]_E(y)=\left\{ 
  \begin{array}{l l}
    (0.001\cdot [\![c]\!]_E+0.2)\cdot 0.4, \,\,\, \, \, \,  \text{if $y=3$}\\
    (0.001\cdot [\![c]\!]_E+0.2)\cdot 0.6, \,\,\, \, \, \,  \text{if $y=4$}\\
    (1-(0.001\cdot [\![c]\!]_E+0.2))\cdot 0.4, \,\,\, \, \, \,  \text{if $y=6$}\\
    (1-(0.001\cdot [\![c]\!]_E+0.2))\cdot 0.6, \,\,\, \, \, \,  \text{if $y=7$}\\
    0, \, \, \, \,\,\, \text{otherwise}\\
  \end{array} \right. \]
\end{example}

The convex combination operator 
is the only one that is not closed with respect to pmfs whose support is a single point.  
{ 
Lemma \ref{conv} shows the associativity of the convex distribution.
\begin{lemma}\label{conv}
Given probability mass functions $\pi_1$, $\pi_2 : \mathbb{N} \rightarrow [0,1]$, $p_1,p_2,p_3,p_4 \in [0,1]$ and $k \in \mathbb{Q}_{\geq 0}$, then the following equations hold:
\begin{itemize}
\item $k(({\pi_1})_p : \pi_2)=(k{\pi_1})_p : (k\pi_2)$
\item $(({\pi_1})_{p_1} : \pi_2)_{p_2}: \pi_3 = ({\pi_1})_{p_3} :( ({\pi_2})_{p_4}:\pi_3)$ iff $p_3=p_1 p_2$ and $ p_4=\frac{(1-p_1)p_2}{1-p_1 p_2}$
\item $({\pi_1})_p : \pi_2=({\pi_2})_{1-p} : \pi_1$
\item $({\pi_1})_p : \pi_1=\pi_1$.
\end{itemize}
\end{lemma}

\begin{proof}
We need to prove each statement.

\par

\noindent
\emph{Case $k(({\pi_1})_p : \pi_2)=(k{\pi_1})_p : (k\pi_2)$.}
\par
For $y \in \mathbb{N}$ we have that 
\begin{align*}
& k({(\pi_1)}_p : \pi_2)(y)= \\
&\sum_{y_i \in \mathbb{N} s.t. \lfloor ky_i \rfloor=y} (p\pi_1 (y_i)+(1-p)(\pi_2(y_i)))=\\
&\sum_{y_i \in \mathbb{N}\, s.t.\, \lfloor ky_i \rfloor=y} (p\pi_1 (y_i))+
\\&\quad \quad \quad \quad \quad \quad\quad \quad  \sum_{y_i \in \mathbb{N}\, s.t. \, \lfloor ky_i \rfloor=y}((1-p)(\pi_2(y_i)))=\\
&p\cdot \sum_{y_i\in \mathbb{N}\, s.t.\, \lfloor ky_i \rfloor=y} (\pi_1 (y_i))+(1-p)\cdot \\
&\quad \quad \quad \quad \quad \quad \quad \quad  \sum_{y_i \in \mathbb{N}\, s.t. \,\lfloor ky_i \rfloor=y}((\pi_2(y_i)))=\\
&(k{\pi_1})_p : (k\pi_2))(y)
\end{align*}

\par

\noindent
\emph{Case $(({\pi_1})_{p_1} : \pi_2)_{p_2}: \pi_3 = ({\pi_1})_{p_3} :( ({\pi_2})_{p_4}:\pi_3)$ iff $p_3=p_1 p_2$ and $ p_4=\frac{(1-p_1)p_2}{1-p_1 p_2}$.}
\par
For $y \in \mathbb{N}$ we have that 
\begin{align*}
(({\pi_1}_{p_1} &: \pi_2)_{p_2}:\pi_3)(y)=\\
&p_2(p_1 \pi_1(y)+(1-p_1)\pi_2(y))+(1-p_2)\pi_3(y)\\
({\pi_1}_{p_3} &:( {\pi_2}_{p_4}:\pi_3))(y)=\\
&p_3\pi_1(y)+ (1-p_3)(p_4 \pi_2(y)+(1-p_4)\pi_3(y))
\end{align*}
These are equal if 
\begin{align*}
& p_1 p_2=p_3\\  
&p_4-p_3 p_4=p_2-p_1p_2 \\
&1-p_2=(1-p_3)(1-p_4)
\end{align*}
and these conditions are satisfied if and only if $p_3=p_1 p_2$ and $ p_4=\frac{(1-p_1)p_2}{1-p_1 p_2}$.

\par

\noindent
\emph{Case $({\pi_1})_p : \pi_2=({\pi_2})_{1-p} : \pi_1$.}
\par
\noindent
For $y \in \mathbb{N} $ by definition \ref{Op-defn} it holds that \[ (({\pi_1})_p : \pi_2)(y)=p \pi_1(y) + (1-p) \pi_2 (y) = \] \[= (1-p)\pi_2 (y)+ p \pi_1 (y)= (({\pi_2})_{1-p} : \pi_1)(y)\]

\par

\noindent
\emph{Case $({\pi_1})_p : \pi_1=\pi_1$.}
\par
For $y \in \mathbb{N} $ by definition \ref{Op-defn} it holds that \begin{align*}
    &(({\pi_1})_p : \pi_1)(y)=p \pi_1(y) + (1-p) \pi_1 (y) =\\
    & (p+1-p)\pi_1 (y)=  \pi_1(y)
\end{align*}  
\hfill $\square$
\end{proof}
}

Having formally defined all the operations on pmfs, we can finally state the following proposition guaranteeing that the semantics of any formula of the  calculus is a pmf.
\begin{proposition}\label{AllIsAPmf}
Given $P$, a formula of the calculus defined in Definition \ref{LangSynt}, and an environment $E$ such that $V(P)\subseteq dom(E)$, then $[\![P]\!]_E$ is a pmf.
\end{proposition}
{ 
\begin{proof}
The proof is by structural induction on the structure of $P$ with basic cases $[\![one]\!]_E=\pi_{one}$ and $[\![zero]\!]_E=\pi_{zero}$, which are  pmfs by definition for any $E$.

\hfill $\square$
\end{proof}
}
The following theorem shows that our calculus is complete with respect to finite support distributions.
\begin{theorem}\label{finite-lang}
For any pmf $f:\mathbb{N}\rightarrow [0,1]$ with finite support there exists a formula $P$ such that $[\![P]\!]_\emptyset=f$.
\end{theorem}

\begin{proof}
Given a pmf $f:\mathbb{N} \rightarrow [0,1]$ with finite support $J=(z_1,...,z_{|J|})$ we can define $P=(z_1\cdot one)_{f(z_1)}:((z_2\cdot one)_{\frac{f(z_2)}{1-f(z_1)}}:(...:((z_i\cdot one)_{\frac{f(z_i)}{\prod_{j=1}^{i-1}(1-f(z_j))}}:...:((z_n\cdot one )))))$.
Then, $[\![P]\!]_\emptyset=f$.
\hfill $\square$
\end{proof}
Proof of Theorem \ref{finite-lang} relies only on a subset of the operators, but the other operators are useful for composing previously defined pmfs.

\section{CRN implementation}\label{CRS imple}
We show how the operators of the calculus can be realized by operators on CRSs. The resulting CRSs produce the required distributions at steady state, that is, in terms of the steady state distribution of the induced CTMC. Thus, we need to consider a restricted class of CRNs that always stabilize and that can be incrementally composed. The key idea is that each such CRN has output species that cannot act as a reactant in any reaction, and hence the counts of those species increase monotonically.\footnote{
Note that this is a stricter requirement than those in \cite{chen2014deterministic}, where output species are produced monotonically, but they are allowed to act as catalysts in some reactions. We cannot allow that because catalyst species influence the value of the propensity rate of a reaction and so the probability that it fires.}
 This implies that the optimized CRSs shown in Section \ref{SpecSect} cannot be used compositionally.
\subsection{Non-reacting output CRSs (NRO-CRSs)}
Since in the calculus presented in Definition \ref{LangSynt} we consider only finite support pmfs, in this section we are limited to finite state CTMCs. This is important because some results valid for finite state CTMCs are not valid in infinite state spaces. 
Moreover, any pmf with infinite support on natural numbers can always be approximated under the $L^1$ norm (see Corollary \ref{univer}).

Given a CRS $C=(\Lambda,R,x_0)$, we call the \emph{non-reacting species} of $C$ the subset of species $\Lambda_r \subseteq \Lambda$ such that given $\lambda_r \in \Lambda_r$ there does not exist $\tau \in R$ such that $r^{\lambda_r}_{\tau}>0$, where $r^{\lambda_r}_{\tau}$ is the component of the source complex of the reaction $\tau$ relative to $\lambda_r$, that is, $\lambda_r$ is not a reactant in any reaction. Given $C$ we also define a subset of species, $\Lambda_o \subseteq \Lambda$, as the \emph{output species} of $C$. Output species are those whose limit distribution is of interest. In general, they may or may not be \emph{non-reacting species}; they depend on the observer and on what he/she is interested in observing. 
\begin{mydef}\label{stable CRN-defn}
A non-reacting output CRS (NRO-CRS) is a tuple $C=(\Lambda,\Lambda_o,R,$ $x_0)$, where $\Lambda_o \subseteq \Lambda$ are the output species of $C$ such that $\Lambda_o \subseteq \Lambda_r$, where $\Lambda_r$ are the non-reacting species of $C$. 
\end{mydef}
NRO-CRNs are CRSs in which the output species are produced monotonically and cannot act as a reactant in any reaction. 
 A consequence of Theorem \ref{th:SingleCase} is the following lemma, which shows that this class of CRNs can approximate any pmf with support on natural numbers, up to an arbitrarily small error.

\begin{lemma}
For any probability mass function $f : \mathbb{N}^{m} \rightarrow [0,1]$ there exists a \emph{NRO-CRS} such that the joint limit distribution of its output species approximates $f$ with arbitrarily small error under the $L^1$ norm. The approximation is exact if the support of $f$ is finite.
\end{lemma}
{ \begin{proof}
This lemma is a consequence of Theorems \ref{th:SingleCase} and  \ref{th:Multi}. In fact, by construction, all CRSs used in those theorems are non-reacting output.
\hfill $\square$
\end{proof}}
%
\subsubsection{NRO-CRS operators}
{  A \emph{NRO-CRS operator} is a NRO-CRS such that, given as input the output of certain NRO-CRSs, it produces as output a (set of) species that at steady state implement a given operation.
We define the following NRO-CRS operators and show their correctness.
\begin{mydef}
\label{operatorsCRS}
Let  $C_1=(\Lambda_1,\Lambda_{o_1},R_1,x_{0_1})$ and $C_2=(\Lambda_2,\Lambda_{o_2},R_2,$ $x_{0_2})$ be NRO-CRSs such that $\Lambda_1 \cap \Lambda_2=\emptyset$. Then, for $\lambda_{o_1}\in \Lambda_{o_1},\lambda_{o_2}\in \Lambda_{o_2}$, $\{\lambda_{out},\lambda_z,\lambda_{r_1},$ $\lambda_{r_2}\}\cap (\Lambda_1\cup \Lambda_2)=\emptyset$, $k\in\mathbf{N},p\in[0,1]$, we define the following NRO-CRS operators:
\begin{align*}
& Sum(C_1,\lambda_{o_1},C_2,\lambda_{o_2},\lambda_{out})=\\&\quad (\Lambda_1\cup \Lambda_2\cup \{ \lambda_{out}\},\{\lambda_{out}\},R_1\cup R_2 \cup \\
& \quad  \{ \lambda_{o_1} \rightarrow \lambda_{out},  \lambda_{o_2} \rightarrow \lambda_{out}\},x_0)\\
& Min(C_1,\lambda_{o_1},C_2,\lambda_{o_2},\lambda_{out})= \\
&\quad (\Lambda_1\cup \Lambda_2\cup \{ \lambda_{out}\},\{\lambda_{out}\},R_1 \cup R_2 \cup\\
&\quad  \{ \lambda_{o_1} +  \lambda_{o_2}\rightarrow \lambda_{out}\},x_0) \\
& Mul(C_1,\lambda_{o_1},k,\lambda_{out})=\\
& \quad  (\Lambda_1 \cup \{\lambda_{out}\},\{\lambda_{out}\},R_1\cup\\ 
&\quad \quad \{\lambda_{o1} \rightarrow \underbrace{\lambda_{out}+...+\lambda_{out}}_{k\, times} \},x_{0}) \\
& Div(C_1,\lambda_{o_1},k,\lambda_{out})=\\
& \quad (\Lambda_1 \cup \{\lambda_{out}\},\{\lambda_{out}\},R_1\cup\\
&\quad  \{ \underbrace{\lambda_{o_1}+...+\lambda_{o_1}}_{k \, times}\rightarrow \lambda_{out} \},x_{0})\\
& Con(C_1,\lambda_{o_1},C_2,\lambda_{o_2},p,\lambda_{out})= \\
& \quad (\Lambda_1\cup \Lambda_2\cup \{\lambda_z,\lambda_{r_1},\lambda_{r_2}, \lambda_{out}\},\{\lambda_{out}\},R_1\cup R_2 \cup\\
&\quad   \,\{\lambda_z \rightarrow^{p} \lambda_{r_1},\lambda_{z} \rightarrow^{1-p} \lambda_{r_2},\\
&\quad  \lambda_{o_1} +\lambda_{r_1} \rightarrow \lambda_{r_1}+\lambda_{out}, \lambda_{o_1} +\lambda_{r_2} \rightarrow \lambda_{r_2}+\lambda_{out}\},x_0)
\end{align*}
where $x_{0}(\lambda)=\begin{cases}
        x_{0_1}(\lambda)&\quad \text{ if $\lambda \in \Lambda_1$}\\
        x_{0_2}(\lambda)&\quad \text{ if $\lambda \in \Lambda_2$}\\
        1&\quad \text{ if $\lambda= \lambda_z$}\\
        0&\quad \text{ otherwise}\\
\end{cases}$
%
\end{mydef}
\begin{theorem}
\label{Th:Operators}
Let  $C_1=(\Lambda_1,\Lambda_{o_1},R_1,x_{0_1})$ and $C_2=(\Lambda_2,$ $\Lambda_{o_2},R_2,$ $x_{0_2})$ be NRO-CRSs such that $\Lambda_1 \cap \Lambda_2=\emptyset$. Then, for $\lambda_{o_1}\in \Lambda_{o_1},\lambda_{o_2}\in \Lambda_{o_2},$ $\lambda_{out}\not\in \Lambda_1\cup \Lambda_2$, $k\in\mathbf{N},p\in[0,1]$ we have:
$$ \pi_{\lambda_{out}}^{Sum(C_1,\lambda_{o_1},C_2,\lambda_{o_2},\lambda_{out})}= \pi_{\lambda_{o_1}}^{C_1}+\pi_{\lambda_{o_2}}^{C_2}$$
$$ \pi_{\lambda_{out}}^{Min(C_1,\lambda_{o_1},C_2,\lambda_{o_2},\lambda_{out})}=min( \pi_{\lambda_{o_1}}^{C_1},\pi_{\lambda_{o_2}}^{C_2}) $$
$$ \pi_{\lambda_{out}}^{Mul(C_1,\lambda_{o_1},k,\lambda_{out})}= k\pi_{\lambda_{o_1}}^{C_1} $$
$$ \pi_{\lambda_{out}}^{ Div(C_1,\lambda_{o_1},k,\lambda_{out})}= \frac{\pi_{\lambda_{o_1}}^{C_1}}{k}$$
$$\pi_{\lambda_{out}}^{Con(C_1,\lambda_{o_1},C_2,\lambda_{o_2},p,\lambda_{out})}= (\pi_{\lambda_{o_1}}^{C_1})_{p}:\pi_{\lambda_{o_2}}^{C_2}  $$
\end{theorem}

} 

{  
\noindent
The proof of Theorem \ref{Th:Operators} is not trivial, and is given in the next subsection. The key difficulties lie in the fact that  we need to compose stochastic processes and show that the resulting process has the required properties.
\begin{example}\label{ex-Sum}
We consider the pmfs $\pi_1$ and $\pi_2$ of Example \ref{ex-COm}. Using the results of Theorem \ref{th:SingleCase} we build the CRSs $C_1$ and $C_2$ such that $\lambda_{out_1}$ and $\lambda_{out_2}$, unique output species of $C_1$ and $C_2$ respectively, admit as steady state distribution exactly $\pi_1$ and $\pi_2$. 
$C_1=(\{\lambda_z,\lambda_1,\lambda_{1,1},\lambda_{o_1}\},$ $\{\lambda_{o_1}\},R,x_0)$ has the following reactions
\[   \lambda_{z} \rightarrow^{\frac{1}{6}} \lambda_{1,1}; \,\,\,\,\,
     \lambda_{z} \rightarrow^{\frac{5}{6}} \emptyset; \,\,\,\,\,
     \lambda_{1}+\lambda_{1,1} \rightarrow^{1} \lambda_{1,1} + \lambda_{o_1}; \,\,\,\,\,
\]
    where $\emptyset$ is the empty set and $x_0$ is such that: $x_0(\lambda_1)=3,\, x_0(\lambda_z)=1, \, x_0(\lambda_{1,1})=0,\, x_0(\lambda_{o_1})=0$.
    
The CRS $C_2$ has the following reactions
\begin{align*}
   &\lambda_{z'} \rightarrow^{\frac{1}{2}} \lambda_{1',1'};\quad  \lambda_{z'} \rightarrow^{\frac{1}{2}} \lambda_{2',2'};\\
   & \lambda_{1'}+\lambda_{1',1'} \rightarrow^{1} \lambda_{1',1'} + \lambda_{o_2};\\ & \lambda_{2'}+\lambda_{2',2'} \rightarrow^{1} \lambda_{2',2'} + \lambda_{o_2};
\end{align*}
with initial condition $x_0$ such that: $x_0(\lambda_{1'})=5,\, x_0(\lambda_z')=1, \, x_0(\lambda_{1',1'})=0, \,x_0(\lambda_{2',2'})=1,\,x_0(\lambda_{2'})=5,\,x_0(\lambda_{o_2})=0$.
Then, applying the Sum operator circuit, we add the following reactions
\[ 
    \lambda_{o_1} \rightarrow^{1} \lambda_{out}; \,\,\,\,\,
     \lambda_{o_2} \rightarrow^{1} \lambda_{out}; 
\]
$Sum(C_1,\lambda_{o_1},C_2,\lambda_{o_2},\lambda_{out})$ has unique output species $\lambda_{out}$, whose limit distribution, $\pi_{\lambda_{out}}$, is equal to $\pi_{1}+\pi_{2}$ described in Example \ref{ex-COm}.
\end{example}
 In what follows, we present in extended form the operator for convex combination, and introduce a new operator, which implements the convex distribution with external inputs ($ConE(\cdot)$). }

Considering $C_1$ and $C_2$, as previously, then we need to derive a CRS operator $Con(C_1,\lambda_{o_1},C_2,\lambda_{o_2},p,\lambda_{out})$ such that $\pi_{\lambda_{out}}=(\pi^{C_1}_{\lambda_{o_1}})_p : (\pi^{C_2}_{\lambda_{o_2}})$. That is, at steady stade, $\lambda_{out}$ equals $\pi^{C_1}_{\lambda_{o_1}}$ with probability $p$ and $\pi^{C_2}_{\lambda_{o_2}}$ with probability $1-p$. This can be done by using Theorem \ref{th:Multi} to generate a bi-dimensional synthetic coin with output species $\lambda_{r_1},\lambda_{r_2}$ such that their joint limit distribution is $$\pi_{\lambda_{r_1},\lambda_{r_2}}(y_1,y_2)=
\begin{cases}
        p       & \quad \text{if $y_1=1$ and $y_2=0$ } \\
        1-p     & \quad \text{if $y_1=0$ and $y_2=1$ } \\
     0    & \quad  \text{otherwise}\\
 \end{cases}.$$ 
That is, $\lambda_{r_1}$ and $\lambda_{r_2}$ are mutually exclusive at steady state. Using these species as catalysts in $\tau_3:\lambda_{o_1}+\lambda_{r_1}\rightarrow\lambda_{r_1}+ \lambda_{out}$ and $\tau_4:\lambda_{o_2}+\lambda_{r_2}\rightarrow\lambda_{r_2}+ \lambda_{out}$ we have exactly the desired result at steady state.
\begin{example}
Consider the following NRO-CRSs $C_1=(\{\lambda_{o_1}\},\{\lambda_{o_1}\},\{\},x_{0_1} )$ and $C_2=(\{\lambda_{o_2}\},\{\lambda_{o_2}\},\{\},x_{0_2} )$, with initial condition $x_{0_1}(\lambda_{o_1})=10$ and $x_{0_2}(\lambda_{o_2})=20$. Then, the operator $Con(C_1,\lambda_{o_1},C_2,\lambda_{o_2},0.3,\lambda_{out})$ implements the operation $\pi_{\lambda_{out}}=(\pi^{C_1}_{\lambda_{o_1}})_{0.3}( \pi^{C_2}_{\lambda_{o_2}})$ and it is given by the following reactions:
\begin{align*}
    &\lambda_z \rightarrow^{0.3} \lambda_{r_1};\quad \lambda_z \rightarrow^{0.7} \lambda_{r_2};\\& \lambda_{r_1} + \lambda_{o_1} \rightarrow \lambda_{r_1}+\lambda_{out};\quad \lambda_{r_2} + \lambda_{o_2} \rightarrow \lambda_{r_2}+\lambda_{out}
\end{align*}
with initial condition $x_0$ such that $x_0(\lambda_z)=1$,  $x_0(\lambda_{r_1})=x_0(\lambda_{r_2})=x_0(\lambda_{out})=0.$
\end{example}
Let $C_1,C_2$ be as above and $f=p_0+p_1\cdot c_1+...+p_n\cdot c_n$ with $p_1,...,p_n \in \mathbb{Q}_{[0,1]}$, $V=\{c_1,...,c_n\}$ a set of environmental variables, and $E$, an environment such that $V\subseteq dom(E)$. Then, computing a CRS operator $ConE(C_1,\lambda_{o_1},C_2,\lambda_{o_2},$ $f(E(V)),\lambda_{out})$ such that $\pi_{\lambda_{out}}=(\pi^{C_1}_{\lambda_{o_1}})_{f(E(V))} : (\pi^{C_2}_{\lambda_{o_2}})$ is a matter of extending the previous circuit.
First of all, we can derive the CRS to compute $f(E(V))$ and $1-f(E(V))$ and memorize them in some species. This can be done as $f(E(V))$ is semi-linear \cite{chen2014deterministic}. Then, as $f(E(V))\leq 1$ by assumption, we can use these species as catalysts to determine the output value of $\lambda_{out}$, as in the previous case. As shown in Sections \ref{Correctness}, this circuit, in the case of external inputs, introduces an arbitrarily small, but non-zero, error, due to the fact that there is no way to know when the computation of $f(E(V))$ terminates. 

\begin{example}
Consider the following NRO-CRSs $C_1=(\{\lambda_{o_1}\},\{\lambda_{o_1} \},\{\},x_{0_1} )$ and $C_2=(\{\lambda_{o_2}\},\{\lambda_{o_2}\},\{\},x_{0_2} )$, with initial condition $x_{0_1}(\lambda_{o_1})=10$ and $x_{0_2}(\lambda_{o_2})=20$. Then, consider the following functions $f(E(c))=E(c)$, where $E$ is a partial function assigning values to $c$, and it is assumed $0.001\leq E(c)\leq 1$ and that $E(c)\cdot 1000 \in \mathbb{N}$. Then, the operator $ConE(C_1,\lambda_{o_1},C_2,\lambda_{o_2},f,\lambda_{out})$, implements the operation $\pi_{\lambda_{out}}=(\pi^{C_1}_{\lambda_{o_1}})_{E(c)}(\pi^{C_2}_{\lambda_{o_2}})$ and it is given by the following reactions:
\begin{align*} 
&\tau_{1}:\lambda_{c} \rightarrow^{k_1} \lambda_{\mathrm{Cat}_1}+\lambda_{\mathrm{Cat}_2};\, \tau_{2}:\lambda_{Tot}+\lambda_{\mathrm{Cat}_2}\rightarrow^{k_1} \emptyset\\
&\tau_{3}:\lambda_{z}+\lambda_{\mathrm{Cat}_1} \rightarrow^{k_2} \lambda_{1};\, \tau_{4}:\lambda_{z}+\lambda_{Tot}\rightarrow^{k_2} \lambda_2\\
&\tau_{5}:\lambda_{o_1}+\lambda_{1} \rightarrow^{k_2} \lambda_{1}+\lambda_{out};\, \tau_{6}:\lambda_{o_2}+\lambda_{2} \rightarrow^{k_2} \lambda_{2}+\lambda_{out}
\end{align*}
where $\lambda_{c},\lambda_{\mathrm{Cat}_1},\lambda_{\mathrm{Cat}_2},\lambda_{z},\lambda_{1}$ and $\lambda_{2}$ are auxiliary species with initial condition $x_0$ such that $x_0(\lambda_{\mathrm{Cat}_1})=x_0(\lambda_{\mathrm{Cat}_2})$ $=x_0(\lambda_{1})=x_0(\lambda_{2})=0,$ $x_0(\lambda_{Tot})=1000, x_0(\lambda_{z})=1$, $x_0(\lambda_{c})=E(c)\cdot1000$ and $k_1 \gg k_2$. Reactions $\tau_1,\tau_2$ implement $f(E(c))$ and $1-f(E(c))$ and store these values in $\lambda_{Cat_1}$ and $\lambda_{Tot}$. These are used in reactions $\tau_3$ and $\tau_4$ to determine the probability that the steady state value of $\lambda_{out}$ is going to be determined by reaction $\tau_5$ or $\tau_6$. 
\end{example}

{ 
\subsection{Correctness of the CRS-operators}
\label{Correctness}
{ We prove the correctness of Theorem \ref{Th:Operators}. For the sake of simplicity, we consider only the Sum operator, as other operators have similar proofs. The key idea of the proof is to make use of Equation \eqref{StochasticSemantics} to show that the resulting CRS implements the desired operation at steady state.
}

\begin{proposition}\label{Sum_Prop}
Let $C_1=(\Lambda_1,\Lambda_{o_1},R_1,x_{0_1}), \, C_2=(\Lambda_2,$ $\Lambda_{o_2},R_2,x_{0_2})$ be NRO-CRSs such that $\Lambda_1 \cap \Lambda_2=\emptyset$ and $\{\lambda_{out}\}\cap (\Lambda_1 \cup \Lambda_2)=\emptyset$. Then for $\lambda_{o_1} \in \Lambda_{o_1}$ and $\lambda_{o_2} \in \Lambda_{o_2}$ the CRS $Sum(C_1,\lambda_{o_1},C_2,\lambda_{o_2},\lambda_{out})=C_c$ is such that $\pi^{C_c}_{\lambda_{out}}= \pi_{\lambda_{o_1}}^{C_1}+\pi_{\lambda_{o_2}}^{C_2}$.
\end{proposition}

  \begin{proof}
Consider  the counting processes $J_{\lambda_{o_1}}^{C_c}$ and $J_{\lambda_{o_2}}^{C_c}$, acording to the stocahstic model introduced in \eqref{StochasticSemantics}, which give the number of molecules of $\lambda_{o_1}$ and $\lambda_{o_2}$ produced until time $t$ in  $C_c$. Using Eqn \eqref{StochasticSemantics} we have
\[
J_{\lambda_{o_1}}^{C_c}(t)=\sum_{\tau \in R_1\cup R_2 \cup\{\tau_{s_1}, \tau_{s_2}\}}p^{\lambda_{o_1}}_{\tau}Y_{\tau}(\int_0^t \! \alpha_{\tau}(X^{C_c}(s)) \, \mathrm{d}s)
\]
\[
J_{\lambda_{o_2}}^{C_c}(t)=\sum_{\tau \in R_1\cup R_2 \cup\{\tau_{s_1}, \tau_{s_2}\}}p^{\lambda_{o_2}}_{\tau}Y_{\tau}(\int_0^t \! \alpha_{\tau}(X^{C_c}(s)) \, \mathrm{d}s)
\]
where $p^{\lambda_{o_1}}_{\tau}$ and $p^{\lambda_{o_2}}_{\tau}$ represent the number of molecules of $\lambda_{o_1}$ and $\lambda_{o_2}$ produced by the occurrence of reaction $\tau$.
Recall that $\tau_{s_1}$ and $\tau_{s_2}$ are such that $\tau_{s_1}: \lambda_{o_1}\to \lambda_{out}$ and $\tau_{s_2}: \lambda_{s_2}\to \lambda_{out}$ and $\Lambda_1 \cap \Lambda_2=\emptyset$. As a consequence,  $p^{\lambda_{o_1}}_{\tau_{s_1}}=p^{\lambda_{o_1}}_{\tau_{s_2}}=p^{\lambda_{o_2}}_{\tau_{s_1}}=p^{\lambda_{o_2}}_{\tau_{s_2}}=0$ and we can write
\begin{align*}
J_{\lambda_{o_1}}^{C_c}(t)=&\sum_{\tau \in R_1\cup R_2 \cup\{\tau_{s_1}, \tau_{s_2}\}}p^{\lambda_{o_1}}_{\tau}Y_{\tau}(\int_0^t \! \alpha_{\tau}(X^{C_c}(s)) \, \mathrm{d}s)=\\ &\sum_{\tau \in R_1}p^{\lambda_{o_1}}_{\tau}Y_{\tau}(\int_0^t \! \alpha_{\tau}(X^{C_c}(s)) \, \mathrm{d}s)
\end{align*}
and 
\begin{align*}
J_{\lambda_{o_2}}^{C_c}&(t)=\\
&\sum_{\tau \in R_1\cup R_2 \cup\{\tau_{s_1}, \tau_{s_2}\}}p^{\lambda_{o_2}}_{\tau}Y_{\tau}(\int_0^t \! \alpha_{\tau}(X^{C_c}(s)) \, \mathrm{d}s)=\\
&\quad \sum_{\tau \in  R_2 }p^{\lambda_{o_2}}_{\tau}Y_{\tau}(\int_0^t \! \alpha_{\tau}(X^{C_c}(s)) \, \mathrm{d}s)
\end{align*} 
Moreover, $r_{\tau_{s_1}}^{\lambda}=p_{\tau_{s_1}}^{\lambda}=r_{\tau_{s_2}}^{\lambda}=p_{\tau_{s_2}}^{\lambda}=0$ for any $\lambda \in \Lambda-\{\lambda_{out},\lambda_{o_1},\lambda_{o_2}\}$, that is, $\tau_{s_1}$ and $\tau_{s_2}$ do not produce or consume any species in $\Lambda-\{\lambda_{out},\lambda_{o_1},\lambda_{o_2}\}$. As a consequence, because $x_0(\lambda)=x_{0_1}(\lambda)$ for all $\lambda \in \Lambda_1-\{\lambda_{o_1}\}$, we have  $\int_0^t \! \alpha_{\tau}(X^{C_c}(s)) \, \mathrm{d}s=\int_0^t \! \alpha_{\tau}(X^{C_1}(s)) \, \mathrm{d}s$ for all $\tau \in R_1$ .
In exactly the same way, it is possible to show that the same relation holds for $\lambda_{o_2}$ with respect to $X^{C_2}$, and as a consequence it is also true that $\int_0^t \! \alpha_{\tau}(X^{C_c}(s)) \, \mathrm{d}s=\int_0^t \! \alpha_{\tau}(X^{C_2}(s)) \, \mathrm{d}s$ for all $\tau \in R_2$. As a result:
\begin{align*}
\sum_{\tau \in R_1}p^{\lambda_{o_1}}_{\tau}Y_{\tau}&(\int_0^t \! \alpha_{\tau}(X^{C_c}(s)) \, \mathrm{d}s)=\\
&\sum_{\tau \in R_1}p^{\lambda_{o_1}}_{\tau}Y_{\tau}(\int_0^t \! \alpha_{\tau}(X^{C_1}(s)) \, \mathrm{d}s)
\end{align*}
\begin{align*}
\sum_{\tau \in  R_2 }p^{\lambda_{o_2}}_{\tau}Y_{\tau}&(\int_0^t \! \alpha_{\tau}(X^{C_c}(s)) \, \mathrm{d}s)=\\
&\sum_{\tau \in  R_2 }p^{\lambda_{o_2}}_{\tau}Y_{\tau}(\int_0^t \! \alpha_{\tau}(X^{C_2}(s)) \, \mathrm{d}s)
\end{align*}
Considering that $\lambda_{o_1}$ is an output species in $C_1$ and $\lambda_{o_2}$ is an output species in $C_2$, that is, NRO-CRSs, then for any $\tau \in R_1$ we have that $\upsilon_{\tau}^{\lambda_{o_1}}=p_{\tau}^{\lambda_{o_1}}$ and for any $\tau \in R_2$ $\upsilon_{\tau}^{\lambda_{o_2}}=p_{\tau}^{\lambda_{o_2}}$. As a consequence:
\begin{align*}
X^{C_1}_{\lambda_{o_1}}(t)=&X^{C_s}_{\lambda_{o_1}}(0)+\sum_{\tau \in R_1}{p}^{\lambda_{o_1}}_{\tau} Y_{\tau}(\int_0^t \! \alpha_{\tau}(X^{C_1}(s)) \, \mathrm{d}s )=\\
&\quad X^{C_1}_{\lambda_{o_1}}(0)+J^{C_c}_{\lambda_{o_1}}(t)
\end{align*}
\begin{align*}
X^{C_2}_{\lambda_{o_2}}(t)=&X^{C_2}_{\lambda_{o_2}}(0)+\sum_{\tau \in R_2}{p}^{\lambda_{o_2}}_{\tau} Y_{\tau}(\int_0^t \! \alpha_{\tau}(X^{C_2}(s)) \, \mathrm{d}s )=\\
&\quad X^{C_2}_{\lambda_{o_2}}(0)+J^{C_c}_{\lambda_{o_2}}(t) 
\end{align*}
According to the fact that in the composed NRO-CRS $\lambda_{out}$ is produced only by $\tau_{s_1}$ and $\tau_{s_2}$ such that $p_{\tau_{s_1}}^{\lambda_{out}}=p_{\tau_{s_2}}^{\lambda_{out}}=1 $, and that $\lambda_{out}$ is not consumed in any reaction, and its initial molecular count is $0$. Then, it is possible to write:
 \begin{align*}
  X^{C_c}_{\lambda_{out}}(t)=&0 + Y_{\tau_{s_1}}(\int_{0}^t \! \alpha_{\tau_c}(X^{C_c}(s))\! \mathrm{d}s)+\\
 &Y_{\tau_{s_2}}(\int_{0}^t \! \alpha_{\tau_c}(X^{C_c}(s))\! \mathrm{d}s)
 \end{align*}
 In the same way we can define the stochastic model for the number of molecules of $\lambda_{o_1}$ or $\lambda_{o_2}$ present in $C_c$ at a given time, as given by the number of molecules produced minus the number of molecules consumed. As $\lambda_{o_1}$ and $\lambda_{o_2}$ are consumed only by $\tau_{s_1}$ and $\tau_{s_2}$, and they are not reactant in any other reaction, we have:
 \begin{align*}
  X^{C_c}_{\lambda_{o_1}+\lambda_{o_2}}(t)=&X^{C_c}_{\lambda_{o_1}}(0) + X^{C_c}_{\lambda_{o_2}}(0) + \\ 
  &\sum_{\tau \in R_1}{p}^{\lambda_{o_1}}_{\tau} Y_{\tau}(\int_0^t \! \alpha_{\tau}(X^{C_1}(s)) \, \mathrm{d}s )+\\
 & \sum_{\tau \in R_2}{p}^{\lambda_{o_2}}_{\tau} Y_{\tau}(\int_0^t \! \alpha_{\tau}(X^{C_2}(s)) \, \mathrm{d}s ) -\\
 & Y_{\tau_{s_1}}(\int_{0}^t \! \alpha_{\tau_c}(X^{C_c}(s))\! \mathrm{d}s)-\\ &Y_{\tau_{s_2}}(\int_{0}^t \! \alpha_{\tau_c}(X^{C_c}(s))\! \mathrm{d}s)= \\  X^{C_1}_{\lambda_{o_1}}(t)+&X^{C_2}_{\lambda_{o_2}}(t)-X^{C_c}_{\lambda_{out}}(t)
 \end{align*}
 because $X^{C_c}_{\lambda_{o_1}}(0)=X^{C_1}_{\lambda_{o_1}}(0)$ and $X^{C_c}_{\lambda_{o_2}}(0)=X^{C_1}_{\lambda_{o_2}}(0)$ by assumption.
 
The set of reachable states from $x_0$ in $X^{C_c}$ is finite because the set of reachable states from $x_{0_1}$ in $X^{C_1}$ and from $x_{0_2}$ in $X^{C_2}$ are finite by assumption and $\tau_c$, in a finite time, can fire only a finite number of times. This implies that $X^{C_c}(t)$ for $t \rightarrow \infty$ will reach a bottom strongly connected component (BSCC) of the underlying graph of the state space, with probability $1$ in finite time, because of a well known result of CTMC theory \cite{kwiatkowska2007stochastic}.
In a BSCC, any pair of configurations $x_1$ and $x_2$ are such that $x_1 \rightarrow^* x_2$ and $x_2 \rightarrow^* x_1$. Therefore, any configuration $x$ in any BSCC reachable by $X^{C_c}$ from $x_0$ is such that $x(\lambda_{o_1})=x_0(\lambda_{o_2})=0$, because in a configuration $x_i$ where $x_i(\lambda_{o_1})>0$ or $x_i(\lambda_{o_2})>0$ it is always possible to reach a configuration $x_j$ where $x_j(\lambda_{o_1})=x_i(\lambda_{o_1})-1$ or $x_j(\lambda_{o_2})=x_i(\lambda_{o_2})-1$ and $x_j(\lambda_{out})=x_i(\lambda_{out})+1$, but then there is no way to reach $x_i$ from $x_j$ because $\lambda_{out}$ is not reactant in any reaction in $R_1\cup R_2 \cup \{\tau_{s_1},\tau_{s_2}\}$.
Therefore  
\[
\lim_{t \rightarrow \infty}Prob(X^{C_c}_{\lambda_{o_1}+\lambda_{o_2}}(t)=0|X^{C_c}(0)=x_0)=1 \implies \]
\begin{align*}
\lim_{t \rightarrow \infty} Prob(& X^{C_1}_{\lambda_{o_1}}(t)+X^{C_2}_{\lambda_{o_2}}(t)-X^{C_c}_{\lambda_{out}}(t) = 0|\\ &X^{C_c}(0)=x_0,X^{C_1}(0)=x_{0_1},\\
&X^{C_2}(0)=x_{0_2})=1 \implies
\end{align*}
\begin{align*}
\lim_{t \rightarrow \infty}Prob(&X^{C_c}_{\lambda_{out}}(t)=X^{C_1}_{\lambda_{o_1}}(t)+X^{C_2}_{\lambda_{o_2}}(t)|\\&X^{C_c}(0)=x_0,X^{C_1}(0)=x_{0_1},\\
&X^{C_2}(0)=x_{0_2})=1
 \end{align*}
 This concludes the proof.
\hfill $\square$
 \end{proof}

}

\subsection{Compiling into the class of NRO-CRSs} 
Given a formula $P$ as defined in Definition \ref{LangSynt}, then $[\![P]\!]_E$ associates to $P$ and an environment $E$ a pmf. We now define a translation of $P$, $T(P)$, into the class of NRO-CRSs 
that guarantees that the unique output species of $T(P)$, at steady state, approximates $[\![P]\!]_E$ with arbitrarily small error for any environment $E$ such that $V(P)\subseteq dom(E)$.
In order to define such a translation we need the following renaming operator.
\begin{mydef}\label{op0-ren}
Given a CRS $C=(\Lambda,R,x_0)$, for $\lambda_t \in \Lambda$ and $\lambda_1 \not\in \Lambda$ we define the renaming operator $C\{\lambda_1 \leftarrow \lambda_t \}=C_c$ such that $C_c=((\Lambda-\{\lambda_t\})\cup\{ \lambda_1\}, R\{\lambda_1 \leftarrow \lambda_t \},x_0')$, where $R\{\lambda_1 \leftarrow \lambda_t \}$  substitutes any occurrence of $\lambda_t$ with an occurrence of $\lambda_1$ for any $\tau \in R$ and $x_0'(\lambda)=\{x_0(\lambda) \,\,\,  \text{if $\lambda \neq \lambda_t$}; \, x_0(\lambda_t) \,\,\, $ $  \text{if $\lambda=\lambda_1$} \} $.

\end{mydef}
This operator produces a new CRS where any occurrence of a species is substituted with an occurrence of another species previously not present. 
\begin{mydef}{(Translation into NRO-CRSs)\label{translation}
Define the mapping $T$ by induction on syntax of formulae $P$:}\label{Trans-DEnsem}
\begin{align*}
T & (one)= (\{\lambda_{out}\},\{\lambda_{out}\},\emptyset,x_0)\,\,\, \text{with $x_0(\lambda_{out})=1$};&&\nonumber \\
T & (zero)= (\{\lambda_{out}\},\{\lambda_{out}\},\emptyset,x_0)\,\,\, \text{with $x_0(\lambda_{out})=0$};&&\nonumber\\
T & (P_1+P_2)=\\
&Sum(T(P_1)\{\lambda_{o_1} \leftarrow \lambda_{out}\},\\
&\quad \lambda_{o_1},T(P_2)\{\lambda_{o_2} \leftarrow \lambda_{out}\},\lambda_{o_2},\lambda_{out});&&\nonumber\\
T & (k\cdot P)=\\
&Div(Mul(T(P)\{\lambda_{o} \leftarrow \lambda_{out}\},\\
&\quad \lambda_{o},k_1,\lambda_{out})\{\lambda_{o'} \leftarrow \lambda_{out}\}),\lambda_{o'},k_2,\lambda_{out}); &&\nonumber\\
T & (min(P_1,P_2)=\\
&Min(T(P_1)\{\lambda_{o_1} \leftarrow \lambda_{out}\},\\
&\lambda_{o_1},T(P_2)\{\lambda_{o_2} \leftarrow \lambda_{out}\},\lambda_{o_2},\lambda_{out});&&\nonumber\\
T & ((P_1)_D:P_2)=&&\nonumber\\
&\left\{ 
 \begin{array}{l l}
    Con(T(P_1)\{\lambda_{o_1} \leftarrow \lambda_{out}\},\lambda_{o_1},T(P_2)\{\lambda_{o_2} \leftarrow \lambda_{out}\},\\
    \quad \quad \quad \lambda_{o_2},D,\lambda_{out}), \quad \quad \text{if $D=p$ }\\
    ConE(T(P_1)\{\lambda_{o_1} \leftarrow \lambda_{out}\},\lambda_{o_1},T(P_2)\{\lambda_{o_2} \leftarrow \lambda_{out}\},\\
 \quad \quad \quad   \lambda_{o_2},D,\lambda_{out}), \quad  \quad \text{\, if $D=p+\sum_{i=1}^{m}p_i\cdot c_i$}\\
  \end{array} \right.   
\end{align*}
  for $m>1$, $k \in \mathbb{Q}_{>0}$, $k_1,k_2 \in \mathbb{N}$ such that $k=\frac{k_1}{k_2}$ and formulae $P_1,P_2$, which are assumed to not contain species $\lambda_{o_1},\lambda_{o_2}$.
\end{mydef}

\begin{example}
Consider the formula $P_1=(one)_{0.001\cdot c + 0.2}$ $(4\cdot one)+(2 \cdot one)_{0.4}(3\cdot one)$ of Example \ref{pmf}, and an environment $E$ such that $0.000125 \leq E(c) \leq 1$ and suppose $E(c) \cdot 800 \in \mathbb{N}$. We show how the translation defined in Definition \ref{translation} produces a NRO-CRS $C$ with output species $\lambda_{out}$ such that $\pi_{\lambda_{out}}=[\![P_1]\!]_E$. Consider the following NRO-CRSs $C_1,C_2,C_3,C_4$ defined as $C_1=(\{\lambda_{c_1}\},\{\lambda_{c_1}\},\{\},x'_0)$ with $x_0(\lambda_{c_1})=1$, $C_2=(\{\lambda_{c_2}\},$ $\{\lambda_{c_2}\},\{\},$ $x_0)$ with $x_0(\lambda_{c_2})=1$,  $C_3=(\{\lambda_{c_3}\},\{\lambda_{c_3}\},\{\},$ $x_0)$ with $x_0(\lambda_{c_3})=1$, and $C_4=(\{\lambda_{c_4}\},\{\lambda_{c_4}\},\{\},$   $x_0)$ with $x_0(\lambda_{c_2})=1$. Then, we have that :
\begin{align*}
C^c_1=&ConE(C_1,\lambda_{c_1},Mul(C_2,\lambda_{c_2},4,\lambda_{out})\{\lambda_{o_2}\leftarrow\lambda_{out}\},\\
&\lambda_{o_2},0.001\cdot c+0.2,\lambda_{out_1})\\
C^c_2=&Con(Mul(C_3,\lambda_{c_3},2,\lambda_{out})\{\lambda_{o_3}\leftarrow\lambda_{out}\},\lambda_{o_3},\\& Mul(C_4,\lambda_{c_4},3,\lambda_{out})\{\lambda_{o_4}\leftarrow\lambda_{out}\},\lambda_{o_4},0.4,\lambda_{out_2})
\end{align*}are such that $\pi_{\lambda_{out_1}}=\left\{ 
  \begin{array}{l l}
    (0.001\cdot [\![c]\!]_E+0.2), \,\,\, \,  \text{if $y=1$}\\
    1-(0.001\cdot [\![c]\!]_E+0.2), \,\,\, \, \, \,\\
    \quad \quad \quad \quad \quad \quad \quad \quad \quad \text{if $y=4$}\\
    0, \, \, \, \,\,\, \text{otherwise}\\
  \end{array} \right.$, and  $\pi_{\lambda_{out_2}}=\left\{ 
  \begin{array}{l l}
    0.4, \,\,\, \, \, \,  \text{if $y=2$}\\
    0.6, \,\,\, \, \, \,  \text{if $y=3$}\\
    0, \, \, \, \,\,\, \text{otherwise}\\
  \end{array} \right.$.
Then, consider the CRS  $C=Sum(C^c_1\{\lambda_{t_1\leftarrow\lambda_{out_1}}\},\lambda_{t_1},$ $C^c_2\{\lambda_{t_2\leftarrow\lambda_{out_2}}\},$ $\lambda_{t_2},\lambda_{out})$ and we have $\pi_{\lambda_{out}}=[\![P_1]\!]_E$ with arbitrarily small error.
The reactions of $C$ are shown below
\begin{align*}
\text{$Mul$ on inputs} \{&\tau_1: \lambda_{C_2}\rightarrow 4\lambda_{o_1};\quad \tau_2: \lambda_{C_3}\rightarrow 2\lambda_{o_2};\\ 
&\tau_3: \lambda_{C_4}\rightarrow 3\lambda_{o_3}. 
\end{align*}
\begin{align*} \text{$C^c_1$} \left\{\begin{array}{l}
\tau_4: \lambda_{env}\rightarrow^{k} \lambda_{cat_1}+\lambda_{cat_2};\\
\tau_5: \lambda_{cat_1} + \lambda_{z} \rightarrow \lambda_{1} \\
\tau_6: \lambda_{cat_2} + \lambda_{tot} \rightarrow^{k} \emptyset;\\ 
\tau_7: \lambda_{tot} + \lambda_{z} \rightarrow \lambda_{2}\quad\quad\quad\quad \quad\quad\quad\quad  \quad\quad\quad\quad\,\\
\tau_8: \lambda_{1} + \lambda_{o_1} \rightarrow \lambda_{o_1} + \lambda_{out_1};\\ \tau_9: \lambda_{2} + \lambda_{o_2} \rightarrow \lambda_{o_2}+ \lambda_{out_1}\end{array} \right.
\end{align*}
\[\text{$C^c_2$} \left\{\begin{array}{l} 
\tau_{10}: \lambda_{z_1}\rightarrow^{0.6} \lambda_{r_1};\\
 \tau_{11}: \lambda_{z_1} \rightarrow^{0.4} \lambda_{r_2}\quad\quad\quad\quad\quad\quad\quad \quad\quad\quad\quad \quad\quad \quad\quad \,\\
\tau_{12}: \lambda_{r_1} + \lambda_{o_3}; \rightarrow \lambda_{r_1}+\lambda_{out_2};\\ \tau_{13}:\lambda_{r_2} + \lambda_{o_4} \rightarrow \lambda_{r_2}+\lambda_{out_2} \end{array} \right. \]
\[\text{$Sum$} \left\{\tau_{14}:\lambda_{out_1}\rightarrow \lambda_{out};\quad \tau_{15}:\lambda_{out_2}\rightarrow \lambda_{out} \quad\quad\quad\quad\quad\quad \quad \quad\quad \, \quad\quad\quad \right. \]
for $k\gg 1$ and initial condition such that $x_0(\lambda_{env})=E(c)\cdot 800$, $x_0(\lambda_{tot})=800$, $x_0(\lambda_{z})=x_0(\lambda_{z_1})=x_0(\lambda_{z_2})=1=x_0(\lambda_{c_1})=x_0(\lambda_{c_2})=x_0(\lambda_{c_3})=x_0(\lambda_{c_4})=1$, and all other species initialized with $0$ molecules.
\end{example}
\begin{proposition}
For any formula $P$ we have that $T(P)$ is a NRO-CRS.
\end{proposition}
\begin{proof}
The proof is by structural induction. The base cases are $T(zero)$ and $T(one)$, which are NRO-CRSs by definition. 
Assuming $T(P_1)$ and $T(P_2)$ are NRO-CRNs then application of operators of $sum$, $Mul$, $Div$, $Min$, $Con$ and $ConE$ on these CRSs produces a NRO-CRNs by definition of the operators.

\hfill $\square$
\end{proof}
Given a formula $P$ and an environment $E$ such that $V(P)\subseteq dom(E)$, the following theorem guarantees the soundness of $T(P)$ with respect to $[\![P]\!]_E$. In order to prove the soundness of our translation we consider the measure of the multiplicative error between two pmfs $f_1$ and $f_2$ with values in $\mathbb{N}^m$, $m>0$ as $e_m(f_1,f_2)=\max_{n \in \mathbb{N}^m}\min(\frac{f_1(n)}{f_2(n)},\frac{f_2(n)}{f_1(n)})$.   
\begin{theorem}{(Soundness)}\label{soundness}
Given a formula $P$ and $\lambda_{out}$, unique output species of $T(P)$, then, for an environment $E$ such that $V(P)\subseteq dom(E)$, it holds that $\pi^{T(P)}_{\lambda_{out}} = [\![P]\!]_E$ with arbitrarily small error under multiplicative error measure.
\end{theorem}
The proof follows by structural induction.
\begin{remark}
A formula $P$ is finite by definition, so Theorem \ref{soundness} is valid because the only production rule which can introduce an error is $(P_1)_D:(P_2)$ in the case $D\neq p_0$, and we can always find reaction rates to make the total probability of error arbitrarily small.
Note that, by using the results of \cite{soloveichik2008computation}, it would also be possible to show that the total error can be kept arbitrarily small, even if a formula is composed from an unbounded number of production rules. This requires small modifications to the ConE operator following ideas in \cite{soloveichik2008computation}.
\end{remark}
Observe that compositional translation, as defined in Definition \ref{Trans-DEnsem}, generally produces more compact CRNs with respect to the direct translation in Theorem \ref{th:SingleCase}, and in both cases the output is non-reacting, so the resulting CRN can be used for composition. For a distribution with support $J$ direct translation yields a CRN with $2|J|$ reactions, whereas, for instance, the support of the sum pmf has the cardinality of the Cartesian product of the supports of the input pmfs. 

\section{Discussion}
Our goal was to explore the capacity of CRNs to compute with distributions. This is an important goal because, when molecular interactions are in low number, as is common in various experimental scenarios \cite{qian2014parallel}, deterministic methods are not accurate, and stochasticity is essential for cellular circuits. Moreover, there is a large body of literature in biology where stochasticity has been shown to be essential 
and not only a nuisance \cite{eldar2010functional}. 
Our 
work is a step forward towards better understanding of molecular computation. 
In this paper we focused on error-free computation for distributions. It would be interesting to understand and characterize what would happen when relaxing this constraint. That is, if we admit a probabilistically (arbitrarily) small error, does the ability of CRNs to compute on distributions increase? {  Another interesting topic to investigate is whether we can relax the constraint that the output species are produced monotonically. In fact, this is a constraint that is generally not present in natural systems where species undergo production and degradation reactions. {More specifically, we require that a CRN will reach a state where no reactions can happen. In terms of sampling from the distribution, this would require sampling an ensemble of cells since sampling a single cell would yield a single state.}}
Also, we would like to address the problem if it is possible to implement distributions in CRNs without leaders (species being present with initial number of molecules equal to $1$) {and without knowing the precise initial number of molecules for each species. Our constructions, except for the uniform distribution, crucially rely on these assumptions, though may be challenging to obtain in technologies such as DNA strand displacement \cite{soloveichik2010dna}. As a consequence, DNA implementation would become easier if these constraints can be removed. However, it is worth noting that, in a practical scenario, leaders can be thought of as single genes or localized structures \cite{qian2014parallel}, and there exist CRN techniques to produce given concentrations independently of initial conditions \cite{shinar2010structural}.}


\bibliographystyle{spmpsci}      
\bibliography{biblio}   

\begin{thebibliography}{10}
\providecommand{\url}[1]{{#1}}
\providecommand{\urlprefix}{URL }
\expandafter\ifx\csname urlstyle\endcsname\relax
  \providecommand{\doi}[1]{DOI~\discretionary{}{}{}#1}\else
  \providecommand{\doi}{DOI~\discretionary{}{}{}\begingroup
  \urlstyle{rm}\Url}\fi

\bibitem{anderson2010product}
Anderson, D.F., Craciun, G., Kurtz, T.G.: {Product-form stationary
  distributions for deficiency zero chemical reaction networks}.
\newblock Bulletin of mathematical biology \textbf{72}(8), 1947--1970 (2010)

\bibitem{anderson2015stochastic}
Anderson, D.F., Kurtz, T.G.: Stochastic analysis of biochemical systems,
  vol.~1.
\newblock Springer

\bibitem{anderson2006environmentally}
Anderson, J.C., Clarke, E.J., Arkin, A.P., Voigt, C.A.: Environmentally
  controlled invasion of cancer cells by engineered bacteria.
\newblock Journal of molecular biology \textbf{355}(4), 619--627 (2006)

\bibitem{angluin2007computational}
Angluin, D., Aspnes, J., Eisenstat, D., Ruppert, E.: {The computational power
  of population protocols}.
\newblock Distributed Computing \textbf{20}(4), 279--304 (2007)

\bibitem{arkin1998stochastic}
Arkin, A., Ross, J., McAdams, H.H.: Stochastic kinetic analysis of
  developmental pathway bifurcation in phage $\lambda$-infected escherichia
  coli cells.
\newblock Genetics \textbf{149}(4), 1633--1648 (1998)

\bibitem{bortolussi2016approximation}
Bortolussi, L., Cardelli, L., Kwiatkowska, M., Laurenti, L.: Approximation of
  probabilistic reachability for chemical reaction networks using the linear
  noise approximation.
\newblock In: International Conference on Quantitative Evaluation of Systems,
  pp. 72--88. Springer (2016)

\bibitem{cardelli2012cell}
Cardelli, L., Csik{\'a}sz-Nagy, A.: The cell cycle switch computes approximate
  majority.
\newblock Scientific reports \textbf{2} (2012)

\bibitem{Cardelli2016}
Cardelli, L., Kwiatkowska, M., Laurenti, L.: Programming discrete distributions
  with chemical reaction networks.
\newblock In: International Conference on DNA-Based Computers, pp. 35--51.
  Springer (2016)

\bibitem{laurenti2015stochastic}
Cardelli, L., Kwiatkowska, M., Laurenti, L.: Stochastic analysis of chemical
  reaction networks using linear noise approximation.
\newblock Biosystems \textbf{149}, 26--33 (2016)

\bibitem{laurenti2016stochastic}
Cardelli, L., Kwiatkowska, M., Laurenti, L.: A stochastic hybrid approximation
  for chemical kinetics based on the linear noise approximation.
\newblock In: Computational Methods in Systems Biology: 14th International
  Conference, CMSB 2016, Cambridge, UK, September 21-23, 2016, Proceedings, pp.
  147--167. Springer (2016)

\bibitem{chen2014deterministic}
Chen, H.L., Doty, D., Soloveichik, D.: Deterministic function computation with
  chemical reaction networks.
\newblock Natural computing \textbf{13}(4), 517--534 (2014)

\bibitem{Chen2013}
Chen, Y.J., Dalchau, N., Srinivas, N., Phillips, A., Cardelli, L., Soloveichik,
  D., Seelig, G.: {Programmable chemical controllers made from DNA}.
\newblock Nature Nanotechnology \textbf{8}(10), 755--762 (2013)

\bibitem{eldar2010functional}
Eldar, A., Elowitz, M.B.: Functional roles for noise in genetic circuits.
\newblock Nature \textbf{467}(7312), 167--173 (2010)

\bibitem{ethier2009markov}
Ethier, S.N., Kurtz, T.G.: {Markov processes: characterization and
  convergence}, vol. 282.
\newblock John Wiley \& Sons (2009)

\bibitem{fett2007synthesizing}
Fett, B., Bruck, J., Riedel, M.D.: Synthesizing stochasticity in biochemical
  systems.
\newblock In: Design Automation Conference, 2007. DAC'07. 44th ACM/IEEE, pp.
  640--645. IEEE (2007)

\bibitem{kwiatkowska2007stochastic}
Kwiatkowska, M., Norman, G., Parker, D.: {Stochastic model checking}.
\newblock In: Formal methods for performance evaluation, pp. 220--270. Springer
  (2007)

\bibitem{losick2008stochasticity}
Losick, R., Desplan, C.: Stochasticity and cell fate.
\newblock Science \textbf{320}(5872), 65--68 (2008)

\bibitem{mardare2016quantitative}
Mardare, R., Panangaden, P., Plotkin, G.: Quantitative algebraic reasoning pp.
  700--709 (2016)

\bibitem{qian2014parallel}
Qian, L., Winfree, E.: {Parallel and scalable computation and spatial dynamics
  with DNA-based chemical reaction networks on a surface}.
\newblock In: {DNA} Computing and Molecular Programming, pp. 114--131. Springer
  (2014)

\bibitem{schmiedel2015microrna}
Schmiedel, J.M., Klemm, S.L., Zheng, Y., Sahay, A., Bl{\"u}thgen, N., Marks,
  D.S., van Oudenaarden, A.: {MicroRNA} control of protein expression noise.
\newblock Science \textbf{348}(6230), 128--132 (2015)

\bibitem{shinar2010structural}
Shinar, G., Feinberg, M.: Structural sources of robustness in biochemical
  reaction networks.
\newblock Science \textbf{327}(5971), 1389--1391 (2010)

\bibitem{soloveichik2008computation}
Soloveichik, D., Cook, M., Winfree, E., Bruck, J.: {Computation with finite
  stochastic chemical reaction networks}.
\newblock natural computing \textbf{7}(4), 615--633 (2008)

\bibitem{soloveichik2010dna}
Soloveichik, D., Seelig, G., Winfree, E.: Dna as a universal substrate for
  chemical kinetics.
\newblock Proceedings of the National Academy of Sciences \textbf{107}(12),
  5393--5398 (2010)

\bibitem{van1992stochastic}
Van~Kampen, N.G.: {Stochastic processes in physics and chemistry}, vol.~1.
\newblock Elsevier (1992)

\end{thebibliography}

\end{document}